\documentclass[11pt]{article}
\usepackage[pass,letterpaper,margin=1in,nohead]{geometry}
\usepackage{fullpage}
\usepackage{authblk}
\usepackage{graphicx}
\usepackage{amsmath}
\usepackage{amsfonts}
\usepackage{amssymb}
\usepackage{amsthm}
\usepackage{mathrsfs}
\usepackage{bm}
\usepackage{mathtools}
\usepackage{enumitem}

\usepackage[title]{appendix}

\usepackage{setspace}
\usepackage{algpseudocode}

\usepackage{thmtools}
\usepackage{thm-restate}

\usepackage{algorithm,setspace}
\usepackage{lipsum}

\usepackage{url}
\usepackage{doi}
\usepackage{hyperref}
\hypersetup{
    colorlinks=true,
    linkcolor=black,
    filecolor=magenta,      
    urlcolor=blue,
    citecolor=black,
    }

\usepackage{caption}
\usepackage{subcaption}
\newtheorem{theorem}{Theorem}
\newtheorem{lemma}[theorem]{Lemma}

\newtheorem{corollary}[theorem]{Corollary}

\newtheorem{example}[theorem]{Example}

\newtheorem*{claim*}{Claim}
\newtheorem*{remark*}{Remark}

\newtheorem*{definition*}{Definition}

\makeatletter

\usepackage{xparse}
\NewDocumentCommand{\lplabel}{o m}{%
  \makebox[0pt][r]{#2\hspace*{4em}}%
  \IfNoValueF{#1}
    {\def\@currentlabel{#2}\ltx@label{#1}}
}

\renewcommand\section{%
  \@startsection{section}{1}
                {\z@}%
                {-3.5ex \@plus -1ex \@minus -.2ex}%
                {2.3ex \@plus.2ex}%
                {\large\bfseries}
}
\renewcommand\subsection{%
  \@startsection{subsection}{2}
                {\z@}%
                {-3.25ex\@plus -1ex \@minus -.2ex}%
                {1sp}
                {\normalsize\bfseries}
}
\renewcommand\subsubsection{%
  \@startsection{subsubsection}{3}
                {\z@}%
                {-3.25ex\@plus -1ex \@minus -.2ex}%
                {1sp}
                {\normalfont\normalsize}
}
\makeatother

\title{{\LARGE\bf  Arboricity games: the core and the nucleolus}\thanks{This is a corrected version of the article, incorporating corrections and clarifications to some details of the proofs and exposition. These revisions do not affect the validity of the main results and conclusions. The Version of Record was published in \href{https://www.springer.com/journal/10107}{Mathematical Programming} \textbf{198}, 1--25 (2023). DOI: \href{https://doi.org/10.1007/s10107-021-01752-w}{10.1007/s10107-021-01752-w}.}}

\author{Han Xiao\thanks{Corresponding author.}~~and Qizhi Fang}

\affil{School of Mathematical Sciences\\Ocean University of China\\Qingdao, China\\\{hxiao, qfang\}@ouc.edu.cn}

\date{}

\begin{document}


\clearpage\maketitle
\thispagestyle{empty} 

\openup 1.2\jot

\begin{abstract}
The arboricity of a graph is the minimum number of forests required to cover all its edges.
In this paper, we examine arboricity from a game-theoretic perspective and investigate cost-sharing in the minimum forest cover problem.
We introduce the arboricity game as a cooperative cost game defined on a graph.
The players are edges, and the cost of each coalition is the arboricity of the subgraph induced by the coalition.
We study properties of the core and propose an efficient algorithm for computing the nucleolus when the core is not empty.
In order to compute the nucleolus in the core, we introduce the prime partition which is built on the lattice of maximum-density edge sets.
The prime partition decomposes the edge set of a graph into a non-prime set and a partially ordered collection of prime sets defined from minimal densest minors and their invariant precedence relation.
Moreover, edges from the same set of the partition always have the same value in a core allocation.
Consequently, when the core is not empty, the prime partition significantly reduces the number of variables and constraints required in the linear programs of Maschler's scheme and allows us to compute the nucleolus in polynomial time.
Besides, the prime partition provides a graph decomposition analogous to the celebrated core decomposition and the density-friendly decomposition, which may be of independent interest.

\hfill





\end{abstract}

\newpage

\section{Introduction}
The arboricity of a graph is the minimum number of forests required to cover all edges of the graph.
Hence arboricity concerns forest cover, a special case of matroid covering. 
Besides, arboricity is a measure of graph density.
A graph with large arboricity always contains a dense subgraph.
By employing the nontrivial interplay between forest cover and graph density under the polyhedral framework, we examine arboricity from a game-theoretic perspective and introduce the so-called arboricity game.
Briefly, the arboricity game is a cooperative cost game defined on a graph,
where the players are edges and the cost of each coalition is the arboricity of the subgraph induced by the coalition.

A central question in cooperative game theory is to distribute the total cost to its participants.
Many solution concepts have been proposed for cost-sharing.
One solution concept is the core, which requires that no coalition benefits by breaking away from the grand coalition.
Another solution concept is the nucleolus, which is the unique solution that lexicographically maximizes the vector of non-decreasingly ordered excesses.
Following the definition, Kopelowitz \cite{Kope67} and Maschler et al. \cite{MPS79} proposed a standard procedure to compute the nucleolus by solving a sequence of linear programs.
However, the size of these linear programs may be exponentially large due to the number of constraints corresponding to all possible coalitions.
Hence it is in general unclear how to apply this procedure.
The first polynomial algorithm for computing the nucleolus was proposed by Megiddo \cite{Megi78} for cooperative cost games defined on directed trees.
Later on, a number of polynomial algorithms were developed for,
e.g., bankruptcy games \cite{AM85}, matching games \cite{BKP12, CLZ12, KP03, KPT20, SR94}, standard tree games \cite{GMOZ96}, airport profit games \cite{BITZ06}, flow games \cite{DFS09}, voting games \cite{EP09}, spanning connectivity games \cite{ALPS09}, shortest path games \cite{BB19}, and network strength games \cite{BB20}.
On the negative side, NP-hardness results for computing the nucleolus were shown for, e.g., minimum spanning tree games \cite{FKK98}, threshold games \cite{EGGW07}, $b$-matching games \cite{KTZ21}, flow games and linear production games \cite{DFS09, FZCD02}.

The main contribution of this paper is twofold.
One contribution is concerned with the arboricity game, where cost-sharing in the minimum forest cover problem is considered.
We study properties of the core and propose an efficient algorithm for computing the nucleolus when the core is nonempty.
Our results are in the same spirit as \cite{BB20, KPT20}, but justifications are different.
The other contribution goes to the prime partition, which is a graph decomposition analogous to the celebrated core decomposition \cite{Seid83} and the density-friendly decomposition \cite{Tatt19, TG15}.
The prime partition is inspired by the principal partition of matroids \cite{CGHL92} and by the graph decompositions developed in \cite{ALPS09, BB20}.
For arboricity games, the prime partition dramatically reduces the size of linear programs involved in Maschler's scheme and enables us to compute the nucleolus in polynomial time.

The rest of this paper is organized as follows.
Section \ref{sec:Preliminaries} introduces relevant concepts.
Section \ref{sec:PolyhedralCombinatorics} reviews some polyhedral results on arboricity.
Section \ref{sec:Core} studies properties of the core.
Section \ref{sec:PrimePartition} is devoted to the prime partition, a graph decomposition of independent interest.
Section \ref{sec:Nucleolus} develops an efficient algorithm for computing the nucleolus.
Section \ref{sec:Conclusion} concludes this paper.

\section{Preliminaries}
\label{sec:Preliminaries}

A \emph{cooperative game} $\Gamma=(N,\gamma)$ consists of a player set $N$ and a characteristic function $\gamma:2^N\rightarrow \mathbb{R}$ with the convention $\gamma (\emptyset)=0$.
We assume that $N\ne\emptyset$ and $\gamma(N)\geq0$.
The player set $N$ is called the \emph{grand coalition}.
Any subset $S$ of $N$ is called a \emph{coalition}.
Given a vector $\boldsymbol{x}\in \mathbb{R}^N$,
we use $x(S)$ to denote $\sum_{i\in S}x_i$ for any $S\subseteq N$.
A vector $\boldsymbol{x}\in \mathbb{R}^N_{\geq 0}$ is called an \emph{allocation} of $\Gamma$ if $x(N)=\gamma (N)$.
The \emph{excess} of a coalition $S$ at an allocation $\boldsymbol{x}$ is defined as $e(S,\boldsymbol{x})=\gamma (S)-x(S)$. 
The \emph{core} of $\Gamma$, denoted by $\mathcal{C}(\Gamma)$, is the set of allocations where all excesses are nonnegative, i.e.,
\begin{equation*}
\mathcal{C}(\Gamma)=\big\{\boldsymbol{x}\in\mathbb{R}^N_{\geq 0} : x(N)=\gamma (N);\, x(S)\leq \gamma (S), \, \forall S\subseteq N \big\}.
\end{equation*}
The \emph{excess vector} $\theta(\boldsymbol{x})$ of an allocation $\boldsymbol{x}$ is the $2^{\lvert N\rvert}-2$ dimensional vector whose components are
the non-trivial excesses $e(S,\boldsymbol{x})$ for $S\in 2^N\backslash \{\emptyset,N\}$ arranged in a non-decreasing order.
The \emph{nucleolus} \cite{Schm69} is the unique allocation $\boldsymbol{x}$ that lexicographically maximizes the excess vector $\theta(\boldsymbol{x})$.
When the core is nonempty, the nucleolus always exists and lies in the core.
For a one-player game, the unique allocation is its nucleolus; assume $\lvert N\rvert\geq2$ when applying the following linear programming scheme.
Moreover, the nucleolus can always be computed with a standard procedure of Maschler et al. \cite{Kope67,MPS79} by recursively solving a sequence of linear programs.
\begin{alignat}{3}
\max\quad & \epsilon &{}& \label{eq:Nucleolus_LP1_0}\\
\lplabel[lp1]{$(LP_1)$}\mbox{s.t.}\quad
 &x(N) = \gamma (N), &\quad &  \label{eq:Nucleolus_LP1_1}\\
 &x(S)+\epsilon \leq \gamma (S), &\quad &\forall~ S\in 2^N\backslash \{\emptyset,N\}, \label{eq:Nucleolus_LP1_2}\\
 &x_i \geq 0, &\quad &\forall~ i\in N. \label{eq:Nucleolus_LP1_3}
\end{alignat}

To compute the nucleolus with Maschler's scheme, first solve linear program $LP_1$ to maximize the minimum excess among all non-trivial coalitions.
For any constant $\epsilon$, let $P_1(\epsilon)$ denote the set of vectors $\boldsymbol{x}\in \mathbb{R}^N$ such that $(\boldsymbol{x},\epsilon)$ satisfies \eqref{eq:Nucleolus_LP1_1}-\eqref{eq:Nucleolus_LP1_3}, i.e., $P_1(\epsilon)$ is the set of allocations whose minimum excess is no less than $\epsilon$.
It follows that $\mathcal{C}(\Gamma)=P_1(0)$.
Let $\epsilon_1$ be the optimal value of $LP_1$.
Then $P_1(\epsilon_1)$ is the set of optimal solutions of $LP_1$, which is also called the \emph{least core} of $\Gamma$.
Thus $\mathcal{C}(\Gamma)\not=\emptyset$ if and only if $\epsilon_1\geq 0$.
For any polyhedron $P\subseteq \mathbb{R}^N$, let $\text{Fix}(P)$ denote the set of coalitions \emph{fixed} by $P$, i.e.,
\begin{equation*}
\text{Fix}(P)=\big\{S\subseteq N : x(S)=y(S), ~\forall~\boldsymbol{x},\boldsymbol{y}\in P \big\}.
\end{equation*}
After solving linear program $LP_r$, let $\epsilon_r$ be the optimal value and $P_r(\epsilon_r)$ be the set of optimal solutions.
Then solve linear program $LP_{r+1}$ to maximize the minimum excess on coalitions that are not fixed by $P_r(\epsilon_r)$.
\begin{alignat}{3}
\max\quad & \epsilon &{}&\\
\lplabel[lp2]{$(LP_{r+1})$}\mbox{s.t.}\quad
 &x(S)+\epsilon \leq \gamma (S), &\quad& \forall~ S\not\in \text{Fix}\big(P_{r}(\epsilon_{r})\big), \label{eq:Nucleolus_LPr_1}\\
 &\boldsymbol{x}\in P_{r}(\epsilon_{r}) \label{eq:Nucleolus_LPr_2}.
\end{alignat}
Clearly, $\epsilon_{r+1}\geq \epsilon_{r}$ and $P_{r+1}(\epsilon_{r+1})\subseteq P_{r}(\epsilon_{r})$.
Moreover, the dimension of $P_{r+1} (\epsilon_{r+1})$ decreases before it collapses to zero.
Hence it takes up to $\lvert N\rvert$ rounds before $P_{r+1} (\epsilon_{r+1})$ becomes a singleton which is exactly the nucleolus.
However, the linear programs involved in Maschler's scheme are usually of exponential size.
Even if linear programs $LP_{1},\ldots,LP_{r}$ have been successfully solved, it may be intractable in polynomial time to determine all coalitions not fixed by $P_{r}(\epsilon_{r})$.
Hence it is in general unclear how to apply Maschler's scheme.
For arboricity games,
we show that the number of variables and constraints required in the successive linear programs of Maschler's scheme can be dramatically reduced,
and, when the core is nonempty, the nucleolus can be determined efficiently in at most two rounds of Maschler's scheme.

We assume that the readers have a moderate familiarity with graph theory.
We first clarify the assumptions, notions and notation used in this paper.
Throughout this paper, we assume that all graphs are loopless but parallel edges are allowed.
The edgeless arboricity game has the unique empty allocation and is handled separately; otherwise, graphs defining the games are assumed to have at least one edge.
We also assume that loops are always removed during edge contraction.
An \emph{image} is a vertex obtained from edge contraction.
A \emph{minor} is a graph obtained from repeated vertex deletion, edge deletion and edge contraction.
Let $G=(V,E)$ be a graph.
We use $c(G)$ to denote the number of components of $G$.
We use $n(G)$ and $m(G)$ to denote the number of vertices and edges in $G$ respectively.
We write $n$ for $n(G)$ and $m$ for $m(G)$ when no ambiguity occurs.
Let $U\subseteq V$ be a set of vertices.
We write $G-U$ for the graph obtained from $G$ by deleting all vertices in $U$.
Let $F\subseteq E$ be a set of edges.
We write $G\slash F$ for the graph obtained from $G$ by contracting all edges in $F$.
Let $H$ be a subgraph of $G$.
We write $G-H$ for $G-V(H)$ and write $G\slash H$ for $G\slash E(H)$.
If $X$ and $Y$ are two sets of vertices,
we use $(X,Y)$ and $m(X,Y)$ to denote the set and the number of crossing edges between $X$ and $Y$ respectively.
If $X$ and $Y$ are two subgraphs, we write $(X,Y)$ for $\big(V(X),V(Y)\big)$.

\section{Polyhedral results on arboricity}
\label{sec:PolyhedralCombinatorics}
This section reviews some polyhedral results on arboricity.
For more details about polyhedral combinatorics, we refer to \cite{Schr86, Schr03}.

Let $G=(V,E)$ be a graph.
A \emph{forest cover} of $G$ is a set of forests that covers all edges of $G$.
The \emph{arboricity} of $G$, denoted by $a(G)$, is the minimum size of a forest cover of $G$.
The arboricity measures how dense a graph is.
The \emph{density} of $G$, denoted by $g(G)$, is the value of $\frac{m(G)}{n(G)-c(G)}$.
Hence $g(G)=\frac{m(G)}{n(G)-1}$ if $G$ is connected.
By convention, the density of an edgeless graph is zero.
Nash-Williams \cite{NW64} showed that the arboricity of a graph is lower bounded by the maximum density of subgraphs.

\begin{theorem}[Nash-Williams \cite{NW64}]
\label{thm:NW}
The edges of a graph $G$ can be covered by $k$ forests if and only if $\max_{H\subseteq G} g(H) \leq k$.
\end{theorem}
The value of $\max_{H\subseteq G} g(H)$ is called the \emph{fractional arboricity} of $G$ and denoted by $a_f(G)$.
Theorem \ref{thm:NW} implies that $a(G)=\lceil a_f(G)\rceil$.
Notice that the fractional arboricity is attained by a connected subgraph.
It follows that the fractional arboricity of $G$ can be computed by
\begin{equation}\label{eq:SimpleFractionalArboricity}
\max_{H\subseteq G,\,n(H)\geq2} \frac{m(H)}{n(H)-1}.
\end{equation}

Let $\mathcal{F}$ denote the set of forests in $G$.
The forests in $\mathcal{F}$ are the independent sets of the graphic matroid with ground set $E$,
and every forest cover of $G$ is essentially an independent set cover of the graphic matroid.
For an edge set $F\subseteq E$, let $r(F)=n(G)-c((V,F))$ denote its rank in this matroid; when the graph needs to be specified, we write $r_G(F)$.
Additionally, the definitions of density and (fractional) arboricity in the forest cover problem agree with the conventional definitions in the matroid covering problem \cite{Edmo65}.
Hence the forest cover problem is a special case of the matroid covering problem.
Notice that the fractional arboricity of a matroid is always equal to the fractional cover number of independent sets \cite{SU97, Schr03}.
It follows that the value of \eqref{eq:SimpleFractionalArboricity} is equal to the optimal value of linear program \eqref{eq:ForestCover0}-\eqref{eq:ForestCover2}
\begin{alignat}{4}
\min\quad & {\displaystyle\sum_{F\in \mathcal{F}} z_F}  &{}&\label{eq:ForestCover0}\\
\mbox{s.t.}\quad
&\sum_{F:e\in F} z_F \geq 1, &\qquad &\forall ~ e &\in E, \label{eq:ForestCover1}\\
&z_F \geq 0, &\qquad &\forall ~ F &\in \mathcal{F}. \label{eq:ForestCover2}
\end{alignat}
and the optimal value of its dual \eqref{eq:ForestCoverDual0}-\eqref{eq:ForestCoverDual2}.
\begin{alignat}{4}
\max\quad & {\displaystyle\sum_{e\in E} x_e}  &{}&\label{eq:ForestCoverDual0}\\
\mbox{s.t.}\quad
&\sum_{e:e\in F} x_e \leq 1, &\qquad& \forall ~ F &\in \mathcal{F}, \label{eq:ForestCoverDual1}\\
&x_e \geq 0, &\qquad& \forall ~ e &\in E. \label{eq:ForestCoverDual2}
\end{alignat}
Notice that \eqref{eq:SimpleFractionalArboricity} can be reformulated as
\begin{equation}\label{eq:WeightedFractionalArboricity}
\max_{H\subseteq G,\,n(H)\geq2} \boldsymbol{1} \cdot \frac{\boldsymbol{\lambda}^{H}}{n(H)-1},
\end{equation}
where $\boldsymbol{1}\in \mathbb{Z}^E$ is an all-one vector and $\boldsymbol{\lambda}^H \in \mathbb{Z}^E$ is the incidence vector of $E(H)$.
For an edge set $F\subseteq E$, we similarly write $\boldsymbol{\lambda}^F$ for its incidence vector.
Moreover, $\frac{\boldsymbol{\lambda}^{H}}{n(H)-1}$ satisfies \eqref{eq:ForestCoverDual1} and \eqref{eq:ForestCoverDual2} for any $H\subseteq G$ with $n(H)\geq2$.
Consequently, a densest subgraph gives an optimal solution of \eqref{eq:ForestCoverDual0}-\eqref{eq:ForestCoverDual2} of the form $\frac{\boldsymbol{\lambda}^H}{n(H)-1}$ in \eqref{eq:WeightedFractionalArboricity},
which leads to the following lemma that will be used in Section \ref{sec:Core}.
In the remainder of this paper, a \emph{densest subgraph} always refers to a \emph{connected} subgraph with the maximum density,
and a \emph{densest minor} always refers to a \emph{connected} minor the density of which is equal to the fractional arboricity of the graph.

\begin{lemma}
\label{thm:DensestSubgraph_IncidenceVector}
The set of optimal solutions of \eqref{eq:ForestCoverDual0}-\eqref{eq:ForestCoverDual2} is the convex hull of the vectors $\frac{\boldsymbol{\lambda}^H}{n(H)-1}$
for every densest subgraph $H$ of $G$.
\end{lemma}

\begin{proof}
Let $\rho=a_f(G)>0$ and let $\boldsymbol{x}$ be an optimal solution of \eqref{eq:ForestCoverDual0}-\eqref{eq:ForestCoverDual2}, so that $x(E)=\rho$.
In this proof, $H$ ranges over all connected subgraphs of $G$ with $n(H)\geq2$.
By Edmonds' description of the forest polytope \cite{Schr03},
\begin{equation*}
\operatorname{conv}\{\boldsymbol{\lambda}^{F}:F\in\mathcal{F}\}
=\big\{\boldsymbol{u}\in\mathbb{R}^E_{\geq0}:u(E(H))\leq n(H)-1,\ \forall H\big\}.
\end{equation*}
Here the rank inequality for any edge set follows by summing the inequalities for its nontrivial connected components.
Since $m(H)\leq\rho[n(H)-1]$ for every $H$, the vector $\boldsymbol{u}^*=\boldsymbol{1}/\rho$ belongs to this polytope.
Moreover, $\boldsymbol{x}\cdot\boldsymbol{u}^*=x(E)/\rho=1$, whereas $x(F)\leq1$ for every forest $F$.
Thus maximizing $\boldsymbol{x}\cdot\boldsymbol{u}$ over the forest polytope has optimal value $1$, attained at $\boldsymbol{u}^*$.
By linear programming duality, there are optimal dual multipliers $\delta_H\geq0$ satisfying
\begin{equation*}
\sum_{H:e\in E(H)}\delta_H\geq x_e\quad(e\in E),
\qquad
\sum_H[n(H)-1]\delta_H=1.
\end{equation*}
Since $u_e^*=1/\rho>0$ for every edge, complementary slackness implies
\begin{equation*}
x_e=\sum_{H:e\in E(H)}\delta_H\quad(e\in E).
\end{equation*}
It also implies that $\delta_H>0$ only if $u^*(E(H))=n(H)-1$, or equivalently $m(H)=\rho[n(H)-1]$.
Hence every subgraph with a positive multiplier is a densest subgraph, and
\begin{equation*}
\boldsymbol{x}=\sum_H\delta_H\boldsymbol{\lambda}^{H}
=\sum_H\delta_H[n(H)-1]\frac{\boldsymbol{\lambda}^{H}}{n(H)-1}.
\end{equation*}
The coefficients $\delta_H[n(H)-1]$ are nonnegative and sum to one, giving the required convex combination.
For the reverse inclusion, every forest has at most $n(H)-1$ edges in $H$, so the normalized incidence vector of any densest subgraph is feasible and has objective value $\rho$.
Every convex combination of these vectors is therefore optimal.
\end{proof}

Lemma \ref{thm:DensestSubgraph_IncidenceVector} suggests that the set of optimal solutions of \eqref{eq:ForestCoverDual0}-\eqref{eq:ForestCoverDual2} is a convex polytope where every extreme point corresponds to a densest subgraph.
By polyhedral theory \cite{Schr86}, all faces of a convex polytope form a partially ordered set under inclusion.
It turns out that all densest subgraphs also form a partially ordered set under inclusion,
which suggests that faces of the optimal solution polytope of \eqref{eq:ForestCoverDual0}-\eqref{eq:ForestCoverDual2} may be related to densest subgraphs.
It is this observation that leads to the graph decomposition in Section \ref{sec:PrimePartition}.

\section{The core and its properties}
\label{sec:Core}
Throughout this paper, we always assume that the underlying graph of arboricity games is connected.
Let $\Gamma_G=(N,\gamma)$ denote the \emph{arboricity game} defined on a graph $G=(V,E)$, where $N=E$ and $\gamma (S)=a(G[S])$ for $S\subseteq N$.
We start with an alternative characterization for the core.

\begin{lemma}
\label{thm:Core}
Let $\Gamma_G=(N,\gamma)$ be an arboricity game and $\mathcal{T}$ be the set of spanning trees in $G$.
Then
\begin{equation}
\label{eq:Core}
  \mathcal{C}(\Gamma_G)=\big\{\boldsymbol{x}\in \mathbb{R}^E_{\geq 0} : x(E)=\gamma (E);\, x(T)\leq 1,\, \forall \:T\in \mathcal{T} \big\}.
\end{equation}
\end{lemma}

\begin{proof}
Denote by $\mathcal{C}'(\Gamma_G)$ the right-hand side of \eqref{eq:Core}.
We first show that $ \mathcal{C}(\Gamma_G)\subseteq \mathcal{C}'(\Gamma_G)$.
Let $\boldsymbol{x}\in \mathcal{C}(\Gamma_G)$.
For any $T\in \mathcal{T}$, we have $x(T)\leq \gamma(T)=1$.
It follows that $\boldsymbol{x}\in \mathcal{C}'(\Gamma_G)$.
Now we show that $\mathcal{C}'(\Gamma_G)\subseteq \mathcal{C}(\Gamma_G)$.
Let $\boldsymbol{x}\in \mathcal{C}'(\Gamma_G)$ and $S\in 2^N\backslash \{\emptyset\}$.
Assume that $\gamma (S)=k$ and $G[S]$ can be covered by $k$ forests $F_1,\ldots,F_k$.
Let $T_i$ be a spanning tree containing $F_i$.
It follows that
\begin{equation*}
x(S)\leq\sum_{i=1}^{k} x(F_i)\leq \sum_{i=1}^{k} x(T_i)\leq k=\gamma (S),
\end{equation*}
which implies that $\boldsymbol{x}\in \mathcal{C}(\Gamma_G)$.
\end{proof}

A necessary and sufficient condition for core nonemptiness follows immediately.
\begin{theorem}
\label{thm:CoreNonempty}
Let $\Gamma_G=(N,\gamma)$ be an arboricity game.
Then $\mathcal{C}(\Gamma_G)\not=\emptyset$ if and only if $a_f(G)=a(G)$.
\end{theorem}

\begin{proof}
Let $\boldsymbol{x}$ be an optimal solution of \eqref{eq:ForestCoverDual0}-\eqref{eq:ForestCoverDual2}.
It follows that $x(E) = a_f(G)\leq a(G)=\gamma (E)$.
Since $\mathcal{T}\subseteq \mathcal{F}$, Lemma \ref{thm:Core} implies that $\boldsymbol{x}\in \mathcal{C}(\Gamma_G)$ if $a_f(G)=a(G)$.

Let $\boldsymbol{x}\in \mathcal{C}(\Gamma_G)$.
Since every forest is a subgraph of a spanning tree, Lemma \ref{thm:Core} implies that $\boldsymbol{x}$ is a feasible solution of \eqref{eq:ForestCoverDual0}-\eqref{eq:ForestCoverDual2}.
It follows that $x(E)\leq a_f(G)\leq a(G)=\gamma (E)$.
Hence $x(E)=\gamma (E)$ implies $a_f(G)=a(G)$.
\end{proof}

Since the arboricity game is a special case of the covering game,
Theorem \ref{thm:CoreNonempty} respects the universal characterization for core nonemptiness of covering games \cite{DIN99}.
The corollary below also follows from the results for covering games.

\begin{corollary}
\label{thm:CoreCorollary}
Let $\Gamma_G=(N,\gamma)$ be an arboricity game.
Then the nonemptiness of $\mathcal{C}(\Gamma_G)$ can be determined in polynomial time.
Moreover, we can decide in polynomial time if a vector belongs to $\mathcal{C} (\Gamma_G)$, and if not, find a separating hyperplane.
\end{corollary}

Theorem \ref{thm:CoreNonempty} implies that, when the core is nonempty, a vector belongs to the core if and only if it is an optimal solution of \eqref{eq:ForestCoverDual0}-\eqref{eq:ForestCoverDual2}.
Lemma \ref{thm:DensestSubgraph_IncidenceVector} suggests that the nonempty core can be characterized by the incidence vectors of densest subgraphs.
Thus we have the following corollary.

\begin{corollary}
\label{thm:CoreCoverHull}
Let $\Gamma_G=(N,\gamma)$ be an arboricity game with a nonempty core.
Then $\mathcal{C}(\Gamma_G)$ is the convex hull of the vectors
$\frac{\boldsymbol{\lambda}^H}{n(H)-1}$
for every densest subgraph $H$ of $G$.
\end{corollary}
Corollary \ref{thm:CoreCoverHull} implies that every core allocation is a convex combination of vectors, each of which is associated with a densest subgraph.
For edges not in any densest subgraph, we have the following corollary.
\begin{corollary}
\label{thm:NonPrimeSet}
Let $\Gamma_G=(N,\gamma)$ be an arboricity game with a nonempty core.
For any $\boldsymbol{x}\in \mathcal{C}(\Gamma_G)$,
$x_e = 0$ if edge $e$ does not belong to any densest subgraph of $G$.
\end{corollary}

It is well known that the nucleolus lies in the core when the core is nonempty.
To compute the nucleolus in the core, we need a better understanding of the core polytope.
Corollary \ref{thm:CoreCoverHull} states that every extreme point of the core polytope is associated with a densest subgraph, which suggests that faces of the core polytope may also be associated with densest subgraphs.
Inspired by the face lattice of convex polytopes \cite{Schr86}, we introduce a graph decomposition built on densest subgraphs, which is crucial for computing the nucleolus in the core of arboricity games.

\section{The prime partition}
\label{sec:PrimePartition}
This section is self-contained and devoted to the prime partition, a graph decomposition analogous to the core decomposition \cite{Seid83} and the density-friendly decomposition \cite{Tatt19, TG15}.
The prime partition is inspired by the face lattice of convex polytopes and built on the lattice of maximum-density edge sets.
An edge set $F\subseteq E$ is called a \emph{maximum-density edge set} if $\lvert F\rvert=a_f(G)r(F)$; in particular, $\emptyset$ is included.
Such an edge set may induce a disconnected graph; each of its nontrivial components is a densest subgraph.
By applying the uncrossing technique \cite{LRS11} to a chain of subgraphs with the maximum density, we introduce the prime partition.
The \emph{prime partition} decomposes the edge set of a graph into a \emph{non-prime set} and a number of \emph{prime sets}.
The \emph{non-prime set} is the set of edges that are not in any densest subgraph.
The \emph{prime sets} are the incremental edge sets of a chain of subgraphs with the maximum density.
In general, there is more than one chain of subgraphs with the maximum density that defines the prime sets.
A partial order can be defined on the prime sets according to the invariant inclusion relation in any chain of subgraphs defining the prime sets.
There are other graph decompositions \cite{ALPS09, BB20} inspired by the face lattice of convex polytopes.
But they are defined on different discrete structures.
The remainder of this section is organized as follows.
In Subsection \ref{sec:MinimalDensestSubgraph}, we investigate properties of minimal densest subgraphs which are basic ingredients of the prime partition.
In Subsection \ref{sec:PrimeSets}, we define prime sets by levels and introduce the non-prime set as a byproduct.
In Subsection \ref{sec:DensestSubgraphDecomposition}, we show that every densest subgraph admits a unique decomposition with prime sets.
In Subsection \ref{sec:AncestorRelation}, we introduce the ancestor relation of prime sets and define a partial order from the ancestor relation.
Throughout this section, we always assume that the graph $G=(V,E)$ is connected and has at least one edge.
Write $\rho=a_f(G)$, and let $\mathcal{L}$ denote the family of maximum-density edge sets.

\subsection{Minimal densest subgraphs}
\label{sec:MinimalDensestSubgraph}

The following properties of minimal densest subgraphs are useful in defining the prime partition.
The first describes the structure of each minimal densest subgraph, while the next two control how different densest subgraphs intersect.

\begin{lemma}[Cut-vertex-free property]
\label{thm:MinimalDensestSubgraphCutVertex}
Let $H$ be a minimal densest subgraph of $G$.
Then $H$ has no cut vertex.
\end{lemma}
\begin{proof}
Assume to the contrary that $v$ is a cut vertex in $H$.
Let $H_1$ and $H_2$ be two subgraphs of $H$ such that $H_1\cup H_2 =H$ and $H_1\cap H_2=\{v\}$.
Since $H$ is a minimal densest subgraph, we have
\begin{equation}
	g(H_i)=\frac{m(H_i)}{n(H_i)-1} < g(H),
\end{equation}
for $i=1,2$.
It follows that
\begin{equation}
g(H)=\frac{m(H)}{n(H)-1}=\frac{m(H_1)+m(H_2)}{[n(H_1)-1]+[n(H_2)-1]}<g(H),
\end{equation}
which is a contradiction.
Hence $H$ has no cut vertex.
\end{proof}

\begin{lemma}[Noncrossing property]
\label{thm:MDS_Noncrossing}
Let $H$ be a minimal densest subgraph of $G$.
For any densest subgraph $K$ of $G$, either $E(H) \subseteq E(K)$ or $E(H) \cap E(K)=\emptyset$.
\end{lemma}

\begin{proof}
When $E(H) \cap E(K)\not=\emptyset$, assume to the contrary that $E(H)\not\subseteq E(K)$.
Let $X=H\cap K$.
Then $X$ is a proper subgraph of $H$ with $E(X)\not=\emptyset$.
On the one hand, we have
\begin{equation}\label{eq:Noncrossing1}
\frac{m(H-X)+m(H-X,X)}{n(H-X)}>g(H).
\end{equation}
Indeed, otherwise 
\begin{equation}
  \begin{split}
   g(X)&\geq\frac{m(X)}{n(X)-1}=\frac{m(H)-m(H-X)-m(H-X,X)}{n(H)-n(H-X)-1}\\
   & \geq \frac{m(H)-n(H-K)\cdot g(H)}{n(H)-n(H-K)-1}=\frac{m(H)-n(H-K)\cdot \frac{m(H)}{n(H)-1}}{n(H)-n(H-K)-1}=g(H),
  \end{split}
\end{equation}
which contradicts the minimality of $H$.
On the other hand, we have
\begin{equation}\label{eq:Noncrossing2}
\frac{m(K-X)+m(K-X,X)+m(X)}{n(K-X)+n(X)-1}=g(K).
\end{equation}
Since $g(H)=g(K)$, \eqref{eq:Noncrossing1} and \eqref{eq:Noncrossing2} imply that
\begin{equation}
g(H\cup K)\geq \frac{[m(H-X)+m(H-X,X)]+[m(K-X)+m(K-X,X)+m(X)]}{n(H-X)+[n(K-X)+n(X)-1]}> g(K),
\end{equation}
which contradicts the maximum density of $K$. 
Hence either $E(H) \subseteq E(K)$ or $E(H) \cap E(K)=\emptyset$.
\end{proof}

Lemma \ref{thm:MDS_Noncrossing} implies that any two minimal densest subgraphs share no common edge, which is the key property for defining prime sets.
We also notice that any two minimal densest subgraphs share at most one common vertex. 
This observation can be generalized to a ``cycle''-free property for minimal densest subgraphs.

\begin{lemma}[``Cycle''-free property]
\label{thm:MDS_CommonVertex}
Let $H_1,\ldots,H_r$ be distinct minimal densest subgraphs of $G$.
Then $\lvert \{v: v\in V(H_i)\cap V(H_j),~i\not=j \}\rvert < r$.
\end{lemma}

\begin{proof}
Assume to the contrary that $\lvert \{v: v\in V(H_i)\cap V(H_j),~i\not=j \}\rvert \geq r$.
Let $H=\cup_{i=1}^{r} H_i$.
Lemma \ref{thm:MDS_Noncrossing} implies that
\begin{equation}
\begin{split}
g(H)
&\geq \frac{\sum_{i=1}^{r} m(H_i)}{\sum_{i=1}^{r} n(H_i)-\lvert \{v: v\in V(H_i)\cap V(H_j),~ i\not=j \} \rvert -1}\\
&> \frac{\sum_{i=1}^{r} m(H_i)}{\sum_{i=1}^{r} [n(H_i)-1]}=g(H_1),
\end{split}
\end{equation}
which contradicts the maximum density of $H_1$.
\end{proof}

To illustrate the ``cycle''-free property, we introduce an auxiliary graph $\mathcal{H}(G)$.
Every vertex $v_H$ in $\mathcal{H}(G)$ is associated with a minimal densest subgraph $H$ of $G$.
Every edge in $\mathcal{H}(G)$ joins two vertices $v_{H_1}$ and $v_{H_2}$ in $\mathcal{H}(G)$ if $H_1$ and $H_2$ share a common vertex.
Lemma \ref{thm:MDS_CommonVertex} implies that if any three minimal densest subgraphs of $G$ share no common vertex, then $\mathcal{H}(G)$ is acyclic.
The ``cycle''-free property will be used repeatedly in our arguments.

To define the prime partition, we have to determine all minimal densest subgraphs.
Gabow \cite[Theorem~5.1, p.~75]{Gabo98} provided an $O(nm\log(n^2/m))$ algorithm for computing both the integral and fractional arboricity of a capacitated graph.
The discussion following that theorem establishes the same bound on an ordinary RAM when the capacities are integral.
We use this computational model throughout the running-time analysis below.
To apply the result to a loopless multigraph, aggregate each parallel class into one edge whose capacity is the number of edges in that class.
This operation preserves the fractional arboricity, since it preserves the total capacity induced by every vertex set.
Let $\bar{m}$ denote the number of edges after aggregation, so $\bar{m}\leq\min\{m,n(n-1)/2\}$.
The aggregation takes $O(n+m)$ time, and the capacitated algorithm takes $O(n\bar{m}\log(n^2/\bar{m}))$ time.
Since $t\log(2+n^2/t)$ is increasing for $t>0$, we have
\begin{equation*}
\bar{m}\log(n^2/\bar{m})
\leq \bar{m}\log(2+n^2/\bar{m})
\leq m\log(2+n^2/m).
\end{equation*}
Hence the fractional arboricity can be computed in $O(nm\log(2+\frac{n^2}{m}))$ time, including aggregation.
Edgeless graphs are handled directly by the convention $a_f(G)=0$.
By employing the algorithm of Gabow, the enumeration of all minimal densest subgraphs can be done in polynomial time.

\begin{lemma}
\label{thm:MDSEnumeration}
All minimal densest subgraphs of $G$ can be enumerated in $O(n^3 m \log(2+\frac{n^2}{m}))$.
\end{lemma}

\begin{proof}
We first show that a minimal densest subgraph of $G$ can be found in $O(n^2 m \log(2+\frac{n^2}{m}))$.
Initially, compute the fractional arboricity of $G$ and let $H=G$.
Compute the fractional arboricity of $H-v$ where $v\in V(H)$.
If $a_f(H-v)=a_f(G)$, then a densest subgraph of $G$ can be found in $H-v$.
Update $H$ with $H-v$ and repeat the process for $H$ until $a_f(H-v)<a_f(G)$ for any $v\in V(H)$.
Then $H$ is a minimal densest subgraph of $G$.
Process the vertices once in any fixed order: a vertex retained because its deletion lowers the fractional arboricity never needs to be tested again, since subsequent deletions cannot increase that value.
Thus it takes $O(n)$ tests before achieving a minimal densest subgraph of $G$.
Each iteration,
which involves computing the fractional arboricity of a subgraph of $G$,
can be done in $O(n m \log(2+\frac{n^2}{m}))$.
Hence a minimal densest subgraph of $G$ can be found in $O(n^2 m \log(2+\frac{n^2}{m}))$.

Now we show that all minimal densest subgraphs of $G$ can be enumerated in $O(n^3 m \log(2+\frac{n^2}{m}))$.
Let $H_1,\ldots,H_k$ be minimal densest subgraphs that have been found in $G$.
Let $G_k=G - \cup_{i=1}^{k} E(H_i)$.
Compute the fractional arboricity of $G_k$.
If $a_f(G_k)=a_f(G)$, then a minimal densest subgraph $H_{k+1}$ of $G$ can be found in $G_k$ and let $G_{k+1}=G - \cup_{i=1}^{k+1} E(H_i)$.
Repeat this process for $G_{k+1}$ until $a_f(G_{k+1})<a_f(G)$.
Then no minimal densest subgraph of $G$ remains in $G_{k+1}$.
By edge-disjointness and the definition of fractional arboricity,
\begin{equation*}
\rho\sum_{i=1}^{k+1}[n(H_i)-1]
=\left\lvert\bigcup_{i=1}^{k+1}E(H_i)\right\rvert\leq\rho(n-1).
\end{equation*}
Since each $n(H_i)-1\geq1$, there are at most $n-1$ minimal densest subgraphs of $G$.
A minimal densest subgraph of $G$ can be found in $O(n^2m \log(2+\frac{n^2}{m}))$.
Therefore, all minimal densest subgraphs of $G$ can be enumerated in $O(n^3 m \log(2+\frac{n^2}{m}))$.
\end{proof}

\subsection{Defining prime sets by levels}
\label{sec:PrimeSets}

Now we define prime sets in the prime partition.
In short, every prime set is the edge set of a minimal densest minor.
Since edge contractions are involved, we introduce prime sets by levels.
A \emph{prime set of level zero} in $G$ is the edge set of a minimal densest subgraph.
By Lemma \ref{thm:MDS_Noncrossing}, prime sets of level zero are well defined.
Moreover, Lemma \ref{thm:MDSEnumeration} implies that all prime sets of level zero can be enumerated efficiently.
To define prime sets of higher levels, we study properties of densest subgraphs under edge contraction.

\begin{lemma}[Density preserving contraction]
\label{thm:ArboricityPreservingContraction}
Let $H$ be a proper densest subgraph of $G$.
Then $g(G\slash H)\leq g(G)$ and $a_f(G\slash H)\leq a_f(G)$.
Moreover, both equalities hold if $g(G)=a_f(G)$.
\end{lemma}

\begin{proof}
Let $\hat{G}=G\slash H$ and $v_H$ be the image of $H$ in $\hat{G}$.
We first prove that $g(\hat{G})\leq g(G)$ and the equality holds if $g(G)=a_f(G)$.
Notice that $m(\hat{G}-v_H)=m(G-H)$, $m(\hat{G}-v_H,v_H)=m(G-H,H)$ and $n(G-H)=n(\hat{G}-v_H)$.
Hence $g(G)\leq g(H)$ implies that
\begin{equation}
\label{eq:DensestSubgraphContraction}
\underbrace{\frac{m(\hat{G}-v_H)+m(\hat{G}-v_H,v_H)}{n(\hat{G}-v_H)}}_{=g(\hat{G})}\leq 
\underbrace{\frac{[m(G-H)+m(G-H,H)]+m(H)}{n(G-H)+[n(H)-1]}}_{=g(G)}\leq
\underbrace{\frac{m(H)}{n(H)-1}}_{=g(H)}.
\end{equation}
Notice that $g(\hat{G})=\frac{m(G)-m(H)}{n(G)-n(H)}$.
Therefore, $g(\hat{G})=g(G)$ if $g(G)=g(H)$.

Now we prove that $a_f(\hat{G})\leq a_f(G)$.
It suffices to show that $g(\hat{G}')\leq a_f(G)$ for any connected induced subgraph $\hat{G}'$ of $\hat{G}$ with at least two vertices. 
When $v_H\not\in V(\hat{G}')$, it is trivial that $g(\hat{G}')\leq a_f(G)$.
Now assume that $v_H\in V(\hat{G}')$.
Then there is an induced subgraph $G'$ of $G$ such that $H\subseteq G'$ and $\hat{G}'=G'\slash H$.
Hence \eqref{eq:DensestSubgraphContraction} implies that $g(\hat{G}')\leq g(G')\leq a_f(G)$.
It follows that $a_f(\hat{G})\leq a_f(G)$.
We have seen that $g(G)=g(H)$ implies $g(\hat{G})=g(G)$.
Therefore, $a_f(\hat{G}) = a_f(G)$ if $g(G)=g(H)$.
\end{proof}

Lemma \ref{thm:ArboricityPreservingContraction} implies that contracting a densest subgraph does not change the fractional arboricity if this subgraph is a proper subgraph of another densest subgraph.
By Lemma \ref{thm:MDS_Noncrossing}, minimal densest subgraphs possess an uncrossing property.
Hence minimal densest subgraphs can be contracted simultaneously.
After contracting all minimal densest subgraphs, if the fractional arboricity remains unchanged, densest subgraphs of the resulting graph are densest minors of the original graph; moreover, minimal densest subgraphs of the resulting graph are used to define prime sets of level one.
The procedure can be repeated to define prime sets of higher levels until the fractional arboricity of the resulting graph decreases.
In the following, we formally define prime sets of higher levels.

Let $\hat{G}^{(0)}=G$ and $\hat{G}^{(k+1)}$ be the graph obtained from $\hat{G}^{(k)}$ by contracting all edges in minimal densest subgraphs of $\hat{G}^{(k)}$, where $k\geq 0$.
If $a_f(\hat{G}^{(k+1)})=a_f(G)$, then a \emph{prime set of level $k+1$} is the edge set of a minimal densest subgraph in $\hat{G}^{(k+1)}$.
Otherwise, there are no further prime sets and the edge set of $\hat{G}^{(k+1)}$ is the \emph{non-prime set}.
Therefore, every prime set is the edge set of a minimal densest minor arising from the density-preserving contractions above, and the non-prime set is the set of edges that are not in any densest subgraph of $G$.
Here and throughout, ``minimal'' in ``minimal densest minor'' means minimal under subgraph inclusion in the current quotient of the density-preserving contraction process.
For simplicity, we introduce some notation for the prime partition of $G$.
For a prime set $P$ of level $k$, we use $n(P)$ to denote the number of vertices in its defining minimal densest subgraph $\hat{G}^{(k)}[P]$.
We use $\mathcal{P}_k$ to denote the collection of all prime sets of level $k$,
use $\mathcal{P}=\cup_{k} \mathcal{P}_k$ to denote the collection of all prime sets,
use $E_0$ to denote the non-prime set,
and use $\mathcal{E}=\mathcal{P}\cup \{E_0\}$ to denote the prime partition.
Figure \ref{fig:PrimePartition} provides an example of the prime partition.
The fractional arboricity of $G$ is $2$. 
There are $4$ prime sets of level zero, $3$ prime sets of level one, and $2$ prime sets of level two.
The non-prime set is empty.

\begin{figure}[hbt!]
  \centering
  \includegraphics[width=.8\textwidth]{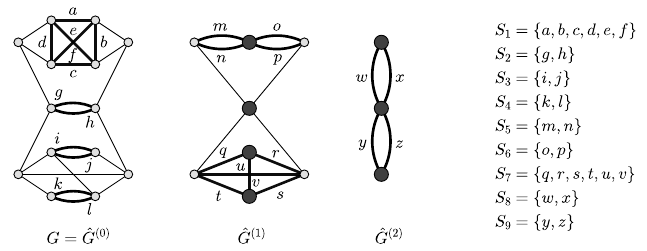}
  \caption{An example of the prime partition.}
  \label{fig:PrimePartition}
\end{figure}

Since the enumeration of minimal densest subgraphs can be done in polynomial time, the prime partition can be computed efficiently.

\begin{theorem}
\label{thm:PrimePartitionComplexity}
The prime partition of $G$ has $O(n)$ prime sets and can be computed in $O(n^4 m \log(2+\frac{n^2}{m}))$.
\end{theorem}

\begin{proof}
Lemmas \ref{thm:MDS_Noncrossing} and \ref{thm:MDS_CommonVertex} imply that contractions of different minimal densest subgraphs can be performed simultaneously.
Notice that $n(\hat{G}^{(k+1)})\leq n(\hat{G}^{(k)})-\lvert \mathcal{P}_{k}\rvert$.
Consequently, there are $O(n)$ prime sets in $\mathcal{P}$.
The definition of prime sets naturally yields an efficient algorithm for computing the prime partition of $G$.
Since there are $O(n)$ prime sets, it takes $O(n)$ iterations to compute the prime partition of $G$, and each iteration computes all prime sets of the same level.
Computing all prime sets of the same level is equivalent to enumerating all minimal densest subgraphs, which can be done in $O(n^3 m \log(2+\frac{n^2}{m}))$.
Hence the prime partition of $G$ can be computed in $O(n^4 m \log(2+\frac{n^2}{m}))$.
\end{proof}

\subsection{Decomposing densest subgraphs with prime sets}
\label{sec:DensestSubgraphDecomposition}

To show that any densest subgraph admits a decomposition of prime sets,
we first record a property of maximum-density edge sets and then generalize the uncrossing property of minimal densest subgraphs to prime sets.

\begin{lemma}
\label{thm:MaximumDensitySets}
The family $\mathcal{L}$ is closed under union and intersection.
For $A,B\in\mathcal{L}$, contracting $A$ removes no edge outside $A$, and
\begin{equation*}
r_{G\slash A}(B\backslash A)=\frac{\lvert B\backslash A\rvert}{\rho}.
\end{equation*}
\end{lemma}

\begin{proof}
For every $F\subseteq E$, we have $\lvert F\rvert\leq\rho r(F)$.
Let $A,B\in\mathcal{L}$.
By the submodularity of the rank function, we have
\begin{equation*}
\begin{split}
\lvert A\rvert+\lvert B\rvert&=\rho[r(A)+r(B)]\\
&\geq\rho[r(A\cup B)+r(A\cap B)]\\
&\geq\lvert A\cup B\rvert+\lvert A\cap B\rvert=\lvert A\rvert+\lvert B\rvert.
\end{split}
\end{equation*}
Hence $A\cup B,A\cap B\in\mathcal{L}$, and equality holds in the rank inequality.
No edge outside $A$ becomes a loop upon contracting $A$: otherwise, for such an edge $e$, we would have $r(A\cup\{e\})=r(A)$ and $\lvert A\cup\{e\}\rvert>\rho r(A\cup\{e\})$.
Consequently, the image of $B$ has edge set $B\backslash A$, and its rank satisfies
\begin{equation*}
\begin{split}
r_{G\slash A}(B\backslash A)
&=r(A\cup B)-r(A)\\
&=r(B)-r(A\cap B)
=\frac{\lvert B\backslash A\rvert}{\rho}.
\end{split}
\end{equation*}
\end{proof}

\begin{lemma}[Generalized noncrossing property]
\label{thm:NonCrossingPrimeSet}
For any prime set $P$ and any densest subgraph $H$ of $G$, either $P\subseteq E(H)$ or $P\cap E(H)=\emptyset$.
\end{lemma}

\begin{proof}
Let $P\in\mathcal{P}_l$.
For $0\leq k\leq l$, let $\hat{H}^{(k)}$ denote the image of $H$ under the contractions producing $\hat{G}^{(k)}$.
We first show that $\hat{H}^{(k)}$ remains a densest subgraph whenever it has an edge, including when prime sets outside $H$ are contracted.
At each level $k<l$, the edge set $A$ contracted in $\hat{G}^{(k)}$ is a union of minimal densest subgraphs.
By Lemma \ref{thm:MaximumDensitySets}, applied in $\hat{G}^{(k)}$, we have $\lvert A\rvert=\rho r_{\hat{G}^{(k)}}(A)$.
The same lemma implies that contracting $A$ preserves the density of the image of $\hat{H}^{(k)}$ whenever that image has an edge; the image remains connected because $H$ is connected.
Starting with $\hat{H}^{(0)}=H$ and applying this argument successively shows that
$g(\hat{H}^{(k)})=\rho=a_f(\hat{G}^{(k)})$ whenever $E(\hat{H}^{(k)})\ne\emptyset$.
If $E(\hat{H}^{(l)})=\emptyset$, then $P\cap E(H)=\emptyset$.
Otherwise, $\hat{H}^{(l)}$ is a densest subgraph of $\hat{G}^{(l)}$.
Since $\hat{G}^{(l)}[P]$ is a minimal densest subgraph of $\hat{G}^{(l)}$,
Lemma \ref{thm:MDS_Noncrossing} implies that either $P\subseteq E(\hat{H}^{(l)})$ or $P\cap E(\hat{H}^{(l)})=\emptyset$.
Therefore, either $P\subseteq E(H)$ or $P\cap E(H)=\emptyset$.
\end{proof}

It follows from Lemma \ref{thm:NonCrossingPrimeSet} that every densest subgraph admits a decomposition of prime sets.

\begin{lemma}[Prime set decomposition]
\label{thm:PrimePartition}
For any densest subgraph $H$ of $G$, there are prime sets $P_1,\ldots,P_r$ such that  $E(H)=\cup_{i=1}^{r} P_{i}$, where $r\geq 1$.
Moreover, $n(H)=\sum_{i=1}^r [n (P_i)-1]+1$.
\end{lemma}

\begin{proof}
Lemmas \ref{thm:ArboricityPreservingContraction} and \ref{thm:NonCrossingPrimeSet} imply that there are prime sets $P_1,\ldots,P_r$ such that $E(H)=\cup_{i=1}^{r} P_i$.
Assume that $P_1,\dots,P_r$ are arranged in a non-decreasing order of levels.
Let $\hat{H}_0=H$ and $\hat{H}_{k}=\hat{H}_{k-1} \slash P_{k}$ for $k=1,\ldots,r$.
By Lemmas \ref{thm:MDS_Noncrossing} and \ref{thm:MDS_CommonVertex},
we have $n(\hat{H}_{k})=n(\hat{H}_{k-1})-n (P_{k})+1$.
Besides, $n(\hat{H}_r)=1$.
Therefore, $n(H)=\sum_{i=1}^{r} [n(\hat{H}_{i-1})-n(\hat{H}_{i})]+n(\hat{H}_r)=\sum_{i=1}^{r} [n (P_i)-1]+1$.
\end{proof}

Since the non-prime set consists of edges that are not in any densest subgraph, 
we have the following lemma.

\begin{lemma}
\label{thm:PrimeComponent}
Let $E_0$ be the non-prime set of $G$.
Then every nontrivial component of $G-E_0$ is a densest subgraph of $G$.
\end{lemma}

\begin{figure}[hbt!]
  \centering
  \includegraphics[width=0.9\textwidth]{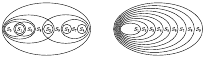}
  \caption{Connected densest subgraphs (left) and a chain of maximum-density edge sets (right) for the graph $G$ in Figure \ref{fig:PrimePartition}.}
  \label{fig:PrimePartitionIllustration}
\end{figure}

The left Venn diagram in Figure \ref{fig:PrimePartitionIllustration} illustrates the relation of all $11$ densest subgraphs of $G$ in Figure \ref{fig:PrimePartition}.
It shows that the edge-set intersection of any two densest subgraphs is either empty or induces a densest subgraph again.
The lattice $\mathcal{L}$ includes disconnected unions of densest subgraphs and $\emptyset$; its closure under union and intersection follows from Lemma \ref{thm:MaximumDensitySets}.
It also shows that every densest subgraph can be decomposed into prime sets. 

\subsection{The ancestor relation}
\label{sec:AncestorRelation}

Lemmas \ref{thm:NonCrossingPrimeSet} and \ref{thm:PrimePartition} suggest that there exists a laminar family of subgraphs with the maximum density that defines the prime sets.
We may determine a chain of subgraphs that defines the prime sets as follows.
Start with $G_0=G-E_0$.
For $k\geq 1$, let $S_k$ be the edge set of a minimal densest subgraph in $G_{k-1}$, $G_{k}$ be the graph obtained from $G_{k-1}$ by contracting edges in $S_k$, and $H_k = G[\cup_{i=1}^{k} S_i]$.
Then $E(H_l)=E\backslash E_0$ for some integer $l$ and $H_k = G[\cup_{i=1}^{k} S_i]$ is a subgraph with the maximum density for $k=1,\ldots,l$. 
Consequently, $\{H_1,\ldots,H_l\}$ with $H_1\subsetneq \ldots \subsetneq H_l$ is a chain of subgraphs with the maximum density that defines the prime sets.
Furthermore, every prime set is precisely the incremental edge set of consecutive subgraphs in the chain.
The right Venn diagram in Figure \ref{fig:PrimePartitionIllustration} provides a chain of subgraphs with the maximum density that defines the prime sets of $G$ in Figure \ref{fig:PrimePartition}.
It shows that all prime sets of $G$ are precisely incremental edge sets of subgraphs in the chain.

Generally, there is more than one chain of subgraphs with the maximum density that defines the prime sets.
However, some prime sets are always preceded by other prime sets in any chain of subgraphs defining the prime sets, as some minimal densest minors occur only after the contraction of other minimal densest minors.
Therefore, we introduce the notion of ancestor to represent the invariant precedence relation in the prime sets.
A prime set $Q$ is called an \emph{ancestor} of a prime set $P$ if the minimal densest minor defining $P$ occurs only after the contraction of $Q$.
Alternatively, $Q$ is an ancestor of $P$ if $Q$ always precedes $P$ in any chain of subgraphs with the maximum density that defines the prime sets.
Clearly, the ancestor relation is \emph{transitive}.

The following lemma shows that the prime sets and the ancestor relation do not depend on the choice of the chain.

\begin{lemma}
\label{thm:PrimeSetChain}
Every maximal sequence of successive contractions of minimal densest subgraphs of density $a_f(G)$ determines the same collection of prime sets.
Moreover, an edge set belongs to $\mathcal{L}$ if and only if it is a union of prime sets containing all ancestors of each of its prime sets.
\end{lemma}

\begin{proof}
By Lemma \ref{thm:MaximumDensitySets}, the family $\mathcal{L}$ is closed under union and intersection.
Every nonempty member of $\mathcal{L}$ is a union of connected densest subgraphs, namely its nontrivial components.

For an edge $e\in E\backslash E_0$, define
\begin{equation*}
C_e=\bigcap\{F\in\mathcal{L}:e\in F\}.
\end{equation*}
Partition $E\backslash E_0$ into classes of edges with identical sets $C_e$.
Equivalently, two edges belong to the same class if and only if they belong to exactly the same members of $\mathcal{L}$.
For a class $P$, write $C(P)=C_e$ for $e\in P$.
Each member of $\mathcal{L}$ is a union of classes, and it contains $C(P)$ whenever it contains $P$.
Conversely, any union of classes having this property belongs to $\mathcal{L}$, because it is the union of the corresponding sets $C(P)$.
The relation $Q\subseteq C(P)$ on classes is reflexive, antisymmetric and transitive.

If $A\in\mathcal{L}$ has already been contracted, the rank of $F\subseteq E\backslash A$ in $G\slash A$ is
\begin{equation*}
r_{G\slash A}(F)=r(A\cup F)-r(A).
\end{equation*}
By Lemma \ref{thm:MaximumDensitySets}, no edge outside $A$ becomes a loop.
Moreover, $\lvert F\rvert=\rho r_{G\slash A}(F)$ if and only if $A\cup F\in\mathcal{L}$.
Thus the minimal nonempty edge sets of density $\rho$ in $G\slash A$ are exactly the uncontracted classes $P$ for which $C(P)\backslash P\subseteq A$.
These edge sets are connected: if such a set had two nontrivial components, each component would itself have density $\rho$.
Consequently, the level construction and every successive minimal-densest-subgraph contraction extract exactly these classes, which are therefore precisely the prime sets.
In particular, a prime set $P$ can be contracted exactly when all prime sets contained in $C(P)\backslash P$ have already been contracted.
It follows that, for distinct prime sets $P,Q$,
\begin{equation*}
Q\text{ is an ancestor of }P
\quad\Longleftrightarrow\quad
Q\subseteq C(P)
\quad\Longleftrightarrow\quad
\bigl(P\subseteq F\Longrightarrow Q\subseteq F\bigr)
\text{ for every }F\in\mathcal{L}.
\end{equation*}
Indeed, any contraction sequence containing $P$ must first contain $C(P)\backslash P$; conversely, the classes in $C(P)$ can be contracted before all classes outside $C(P)$.
Hence the ancestor relation is well defined, and the members of $\mathcal{L}$ are exactly the unions of prime sets that contain every ancestor of each of their prime sets.
\end{proof}

The following consequence of Lemma \ref{thm:PrimeSetChain} will be used in Section \ref{sec:Nucleolus}.

\begin{lemma}
\label{thm:AncestorInclusion}
Let $H$ be a densest subgraph of $G$.
Let $P$ be a prime set of $G$ and $Q$ be an ancestor of $P$.
Then $P\subseteq E(H)$ implies $Q\subseteq E(H)$.
\end{lemma}

Moreover, the ancestor relation can be determined efficiently.

\begin{lemma}
\label{thm:PrimeSetRelation}
Given the prime partition of $G$,
the ancestor relation can be determined in $O(n^2 m)$.
\end{lemma}
\begin{proof}
For a prime set $R$, write $b_R=n(R)-1$, so $\lvert R\rvert=\rho b_R$.
Fix a prime set $Q$ and start with $H=G-Q$ and $A=\emptyset$.
Visit all prime sets $R\ne Q$ in non-decreasing order of levels.
At each visit, compute the graphic rank $r_H(R)$ of the subgraph consisting of the edges in $R$ in the current graph $H$.
If $r_H(R)=b_R$, contract $R$ in $H$ and set $A\gets A\cup R$; otherwise leave $R$ uncontracted.

Throughout the procedure, $A\in\mathcal{L}$, and
\begin{equation*}
r_H(R)=r(A\cup R)-r(A).
\end{equation*}
Indeed, if $A\in\mathcal{L}$, then the test $r_H(R)=b_R$ is equivalent to $A\cup R\in\mathcal{L}$.
By Lemma \ref{thm:PrimeSetChain}, this is equivalent to all ancestors of $R$ already being contained in $A$; when the test succeeds, $R$ is a minimal nonempty edge set of density $\rho$ in the current graph and hence defines a minimal densest subgraph.
Every ancestor of $R$ has a lower level than $R$.
Induction along the visits now shows that a prime set $R\ne Q$ is contracted if and only if $Q$ is not an ancestor of $R$.
If $Q$ is an ancestor, its absence prevents the test from succeeding.
If $Q$ is not an ancestor, then it is not an ancestor of any ancestor of $R$, so all ancestors of $R$ have already been contracted and the test succeeds.
Thus the prime sets left uncontracted, other than $Q$, are exactly those having $Q$ as an ancestor.

There are $O(n)$ prime sets.
For each fixed $Q$, there are $O(n)$ rank computations and contractions, each implementable in $O(m)$ time, giving $O(nm)$ time.
Repeating the procedure for every $Q$ determines the entire ancestor relation in $O(n^2m)$ time.
\end{proof}

We conclude this section with a partially ordered set defined from the prime sets and the ancestor relation.
Write $P\prec Q$ for distinct prime sets $P$ and $Q$ if $Q$ is an ancestor of $P$.
Thus $\prec$ is a strict partial order; in particular, $P\not\prec P$.

\section{Computing the nucleolus}
\label{sec:Nucleolus}

In this section, we develop an efficient algorithm for computing the nucleolus of arboricity games when the core is not empty.
In Subsection \ref{sec:Reformulation}, we employ the prime partition of the underlying graph to reformulate linear programs involved in Maschler's scheme.
In Subsection \ref{sec:Equivalence}, we prove the correctness of our formulation for Maschler's scheme.
In Subsection \ref{sec:CombinatorialAlgorithm}, we show that Maschler's scheme terminates after at most two rounds and the nucleolus can be computed in polynomial time.
Throughout this section, in addition to assuming that graph $G=(V,E)$ is connected, we further assume that arboricity game $\Gamma_G=(N,\gamma)$ has a nonempty core.
If $G$ is a tree with $m\geq1$ edges, every nonempty coalition has cost one and the nucleolus is $x_e=1/m$ for all $e$.
For $m\geq2$, this follows because the minimum excess is $\min_e x_e$, which is uniquely maximized by the uniform allocation; the one-edge case has a unique allocation.
Henceforth in this section, assume $a(G)\geq2$, so every spanning tree is a proper coalition.

\subsection{Reformulating Maschler's scheme}
\label{sec:Reformulation}

To compute the nucleolus of $\Gamma_G$, the first round of Maschler's scheme is to solve linear program $LP_1$ \eqref{eq:Nucleolus_LP1_0}-\eqref{eq:Nucleolus_LP1_3} defined from the standard characterization for the core.
By referring to the alternative characterization for the core in Lemma \ref{thm:Core},
we introduce linear program $LP'_1$ \eqref{eq:Nucleolus_LP1_New_0}-\eqref{eq:Nucleolus_LP1_New_3}.
For any constant $\epsilon$, let $P'_1(\epsilon)$ denote the set of vectors $\boldsymbol{x}\in \mathbb{R}^E$ such that $(\boldsymbol{x},\epsilon)$ satisfies \eqref{eq:Nucleolus_LP1_New_1}-\eqref{eq:Nucleolus_LP1_New_3}.
We show that $LP_1$ and $LP'_1$ are equivalent.
\begin{alignat}{3}
\max\quad & \epsilon &{}& \label{eq:Nucleolus_LP1_New_0}\\
\lplabel[lp1_v2]{$(LP'_1)$\quad}\mbox{s.t.}\quad
 &x(E) = \gamma (E), &\quad &  \label{eq:Nucleolus_LP1_New_1}\\
 &x(T) + \epsilon\leq 1, &\quad &\forall~ T\in \mathcal{T}, \label{eq:Nucleolus_LP1_New_2}\\
 &x_e \geq 0, &\quad &\forall~ e\in E. \label{eq:Nucleolus_LP1_New_3}
\end{alignat}

\begin{lemma}
\label{thm:LeastCore_Alternative}
Let $\epsilon_1$ and $\epsilon'_1$ be the optimal values of $LP_1$ and $LP'_1$, respectively.
Then $\epsilon_1=\epsilon'_1$ and $P_1(\epsilon_1)=P'_1(\epsilon'_1)$.
\end{lemma}

\begin{proof}
We first show that $\epsilon_1=\epsilon'_1$.
It is easy to see that $\epsilon_1\leq \epsilon'_1$, since $LP'_1$ is a relaxation of $LP_1$.
Let $\emptyset\ne S\subsetneq E$ and $\mathcal{C}_S$ be an edge-disjoint minimum forest cover in $G[S]$.
For any $\boldsymbol{x}\in P'_1(\epsilon'_1)$, Lemma \ref{thm:Core} implies that
\begin{equation}
\gamma (S)-x(S) = \sum_{T_S\in \mathcal{C}_S} 1 -\sum_{e\in S} x_e = \sum_{T_S\in \mathcal{C}_S} (1-\sum_{e\in T_S} x_e) \geq \sum_{T_S\in \mathcal{C}_S} \epsilon'_1\geq \epsilon'_1.
\end{equation}
We remark that the last inequality follows from the assumption that $\mathcal{C}(\Gamma_G)\not=\emptyset$ which implies $\epsilon_1\geq 0$.
By the optimality of $\epsilon_1$, we have $\epsilon_1\geq \epsilon'_1$.
Thus $\epsilon_1=\epsilon'_1$ follows.

Next we show that $P_1(\epsilon_1)=P'_1(\epsilon'_1)$.
Clearly, $P_1(\epsilon_1)\subseteq P'_1 (\epsilon'_1)$, since $\epsilon_1=\epsilon'_1$ and  $LP'_1$ is a relaxation of $LP_1$.
For any $\boldsymbol{x}\in P'_1 (\epsilon'_1)$, $\boldsymbol{x}$ also satisfies the constraints of $LP_1$ and attains the optimum.
Hence $\boldsymbol{x}\in P_1 (\epsilon_1)$, which implies that $P'_1 (\epsilon'_1)\subseteq P_1 (\epsilon_1)$.
Thus, $P_1(\epsilon_1)=P'_1(\epsilon'_1)$.
\end{proof}

Before proceeding to the second round of Maschler's scheme, we have to determine the optimal value $\epsilon'_1$ of $LP'_1$.
Clearly, $\epsilon'_1 \geq 0$ as $\mathcal{C}(\Gamma_G)\not=\emptyset$.
Assume that $\gamma(E)=k$ and that $G$ can be covered by $k$ edge-disjoint forests $F_1,\ldots,F_k$.
Let $T_i$ be a spanning tree containing $F_i$ and $\boldsymbol{x}\in \mathcal{C}(\Gamma_G)$.
Clearly, $x(F_i)\leq x(T_i)\leq \gamma(T_i)=1$.
It follows that
$x(E)=\sum_{i=1}^{k} x(F_i)\leq \sum_{i=1}^{k} x(T_i)\leq k=\gamma (E)$.
Then $x(E)=\gamma(E)$ implies $x(F_i)=x(T_i)=1$.
Hence $\epsilon'_1=0$, implying that the core $P'_1(0)$ and the least core $P'_1 (\epsilon'_1)$ coincide.
Consequently, there are spanning trees $T\in \mathcal{T}$ such that $x(T)=1$ for any $\boldsymbol{x}\in \mathcal{C}(\Gamma_G)$, i.e., $T$ is fixed by $P'_1 (\epsilon'_1)$.
Denote by $\mathcal{T}_0$ the set of spanning trees satisfying $x(T)=1$ for every $\boldsymbol{x}\in\mathcal{C}(\Gamma_G)$.
Let $E_0$ denote the set of edges that are not in any densest subgraph of $G$.
Corollary \ref{thm:NonPrimeSet} implies that $x_e=\epsilon'_1=0$ for any $e\in E_0$.
If the core is a singleton, its unique allocation is the nucleolus and no further round is needed.
For the second-stage argument, assume henceforth that the core is not a singleton.
We now give an equivalent formulation $LP'_2$ for the second round of Maschler's scheme.
\begin{alignat}{3}
\max\quad  &{\epsilon}  &{}& \label{eq:Nucleolus_LP2v2_0}\\
\mbox{s.t.}\quad
&x(T)+\epsilon \leq 1, &\qquad& \forall ~ T \in \mathcal{T}\backslash \mathcal{T}_0, \label{eq:Nucleolus_LP2v2_1}\\
\lplabel[lp2_v2]{$(LP'_2)$\quad\quad}
&x(T) = 1, &\qquad& \forall ~T \in \mathcal{T}_0, \label{eq:Nucleolus_LP2v2_2}\\
&x_e \geq \epsilon, &\qquad& \forall ~ e\in E\backslash E_0, \label{eq:Nucleolus_LP2v2_3}\\ 
&x_e = 0, &\qquad& \forall ~ e\in E_0, \label{eq:Nucleolus_LP2v2_4}\\
&x(E)=\gamma(E), &\qquad& \label{eq:Nucleolus_LP2v2_efficiency}\\
&\epsilon\geq0. &\qquad& \label{eq:Nucleolus_LP2v2_nonnegative}
\end{alignat}

\begin{lemma}
\label{thm:SecondRoundFormulation}
When the core is not a singleton, $LP_2$ and $LP'_2$ have the same optimal value and the same optimal allocations.
\end{lemma}

\begin{proof}
We show that the feasible solutions of the two programs coincide for $\epsilon\geq0$.
First, $x_e$ is nonconstant on the core for every $e\notin E_0$.
Indeed, if $x_e$ were a positive constant, Corollary~\ref{thm:CoreCoverHull} would imply that every densest subgraph contains $e$ and has the same rank.
By Lemma \ref{thm:MaximumDensitySets}, their connected union is also densest, so all these subgraphs have the same vertex set and hence the same edge set, since densest subgraphs are vertex-induced.
This would make the core a singleton, a contradiction.
Moreover, $\gamma(E\backslash\{e\})=\gamma(E)$ for every such $e$: otherwise integrality of the costs and the core inequalities would give $1\geq x_e\geq\gamma(E)-\gamma(E\backslash\{e\})\geq1$, contradicting nonconstancy.
Thus the corresponding complementary coalitions, which are not fixed by the core, give $x_e\geq\epsilon$ in $LP_2$.
The constraints for spanning trees not fixed by the core also occur in $LP_2$.
If a spanning tree outside $\mathcal{T}_0$ is fixed at a value $c<1$, then for every densest subgraph $H$ and every $e\notin E_0$, the vector $\boldsymbol{x}^H=\boldsymbol{\lambda}^H/(n(H)-1)$ satisfies
\begin{equation*}
1-c=1-\frac{\lvert T\cap E(H)\rvert}{n(H)-1}\geq\frac{1}{n(H)-1}\geq x^H_e.
\end{equation*}
Taking convex combinations gives $1-x(T)\geq x_e\geq\epsilon$, so its constraint is redundant.

Conversely, a feasible solution of $LP'_2$ belongs to the core by Lemma~\ref{thm:Core}.
For a coalition $S$ not fixed by the core, let $F_1,\ldots,F_q$ be an edge-disjoint minimum forest cover of $G[S]$ and extend each $F_i$ to a spanning tree $T_i$.
Then
\begin{equation*}
\gamma(S)-x(S)=\sum_{i=1}^q[1-x(T_i)]+\sum_{i=1}^q x(T_i\backslash F_i).
\end{equation*}
All summands are nonnegative on the core.
If all $T_i$ were in $\mathcal{T}_0$ and all edges in $T_i\backslash F_i$ were in $E_0$, this excess would be identically zero, contrary to $S$ not being fixed by the core.
Therefore at least one summand is at least $\epsilon$ by $LP'_2$, proving the required constraint of $LP_2$.
Since the second-stage optimum is nonnegative, imposing $\epsilon\geq0$ does not change it.
\end{proof}

However, $LP'_2$ still has an exponential number of constraints. 
We derive an equivalent formulation of $LP'_2$ that has only polynomial size by resorting to the prime partition of $G$.
Notice that $E_0$ in \eqref{eq:Nucleolus_LP2v2_4} is precisely the non-prime set of $G$.
Let $\mathcal{P}=\cup_{k} \mathcal{P}_k$ denote the collection of all prime sets of $G$, where $\mathcal{P}_k$ is the collection of all prime sets of level $k$.
Let $\mathcal{E}=\mathcal{P}\cup \{E_0\}$ denote the prime partition of $G$.
Corollary \ref{thm:NonPrimeSet} states that all edges in the non-prime set have the same value in a core allocation.
It turns out that this property also holds for edges from the same prime set.

\begin{lemma}
\label{thm:PrimeSet_PartialOrder}
Let $\boldsymbol{x}$ be a core allocation of $\Gamma_G$ and $P$ be a prime set of $G$.
Then $x_e=x_f$ for any $e,f\in P$. 
\end{lemma}
\begin{proof}
Let $\boldsymbol{x}^H$ be the vector associated with a densest subgraph $H$ of $G$.
By Lemma \ref{thm:NonCrossingPrimeSet}, either $P\subseteq E(H)$ or $P\cap E(H)=\emptyset$.
Thus for any $e,f\in P$, $x^H_e=x^H_f=\frac{1}{n(H)-1}$ if $P\subseteq E(H)$ and $x^H _e=x^H_f=0$ otherwise.
Since any vector $\boldsymbol{x}\in \mathcal{C}(\Gamma_G)$ is a convex combination of vectors associated with a densest subgraph of $G$, we have $x_e=x_f$ for any $e,f\in P$.
\end{proof}

Corollary \ref{thm:NonPrimeSet} and Lemma \ref{thm:PrimeSet_PartialOrder} state that 
all edges in the same set of the prime partition have the same value in a core allocation.
Hence every core allocation $\boldsymbol{x}\in \mathbb{R}^{E}$ of $\Gamma_G$ defines a vector $\boldsymbol{y}\in \mathbb{R}^{\mathcal{E}}$ associated with the prime partition of $G$.
Moreover, $LP'_2$ can be reformulated with $\boldsymbol{y}$.
Let $(\mathcal{P},\prec)$ denote the partially ordered set defined on $\mathcal{P}$ from the ancestor relation.
Let $\mathcal{P}_{\min}$ denote the set of minimal prime sets in $(\mathcal{P},\prec)$.
Denote by $LP''_2$ the linear program \eqref{eq:Nucleolus_LP2v3_0}-\eqref{eq:Nucleolus_LP2v3_nonnegative} defined on $\boldsymbol{y}$.
In Subsection \ref{sec:Equivalence}, we show that $LP'_2$ and $LP''_2$ are equivalent, i.e., $\boldsymbol{x}$ is feasible to $LP'_2$ if and only if $\boldsymbol{y}$ is feasible to $LP''_2$.
In Subsection \ref{sec:CombinatorialAlgorithm}, we propose a combinatorial algorithm for $LP''_2$ and show that $LP''_2$ has a unique optimal solution which yields the nucleolus of $\Gamma_G$.
\begin{alignat}{3}
\hspace{4.5em}\max\quad  &{\epsilon}  &{}&\label{eq:Nucleolus_LP2v3_0}\\
\mbox{s.t.}\quad
& y_P  +\epsilon\leq y_Q, &\quad& \forall ~ P\prec Q, \label{eq:Nucleolus_LP2v3_1}\\
\lplabel[lp]{$(LP''_2)$\quad~~\,}
&\textstyle\sum_{P\in \mathcal{P}} [n (P)-1] y_P = 1, &\qquad& \label{eq:Nucleolus_LP2v3_2}\\
& y_P \geq \epsilon, &\quad& \forall ~P\in \mathcal{P}_{\min},\label{eq:Nucleolus_LP2v3_3}\\
& y_{E_0}=0, \label{eq:Nucleolus_LP2v3_4}\\
& \epsilon\geq0. \label{eq:Nucleolus_LP2v3_nonnegative}
\end{alignat}

We first show that every feasible solution of $LP''_2$ defines a core allocation.
To this end, we establish the following inequality.

\begin{lemma}
\label{thm:OrderConeBound}
Let $\boldsymbol{w}\in\mathbb{R}^{\mathcal{E}}$ satisfy $w_{E_0}=0$, $w_P\geq0$ for all prime sets, and $w_P\leq w_Q$ whenever $P\prec Q$.
Give each edge in a set of the prime partition the weight of that set.
Then every spanning tree $T$ satisfies
\begin{equation}\label{eq:OrderConeBound}
w(T)\leq\sum_{P\in\mathcal{P}}[n(P)-1]w_P.
\end{equation}
In particular, every feasible solution of $LP''_2$ defines a core allocation.
\end{lemma}

\begin{proof}
If all weights $w_P$ are zero, \eqref{eq:OrderConeBound} is immediate.
Otherwise, let $0=t_0<t_1<\cdots<t_q$, where $t_1,\ldots,t_q$ are the distinct positive values among the weights $w_P$.
For $i=1,\ldots,q$, let $A_i=\bigcup\{P:w_P\geq t_i\}$.
Each $A_i$ contains all ancestors of each of its prime sets, so Lemma \ref{thm:PrimeSetChain} implies that $A_i\in\mathcal{L}$.
Since $\lvert P\rvert=a_f(G)[n(P)-1]$ for each prime set, we have
$r(A_i)=\sum_{P:w_P\geq t_i}[n(P)-1]$.
Consequently, every spanning tree $T$ satisfies
\begin{equation*}
\begin{split}
w(T)&=\sum_{i=1}^{q}(t_i-t_{i-1})\lvert T\cap A_i\rvert \leq\sum_{i=1}^{q}(t_i-t_{i-1})r(A_i)
=\sum_{P\in\mathcal{P}}[n(P)-1]w_P.
\end{split}
\end{equation*}
Here, for each prime set $P$, the differences $t_i-t_{i-1}$ with $t_i\leq w_P$ sum to $w_P$.
A feasible vector of $LP''_2$ has these properties: every prime set lies above a minimal prime set, so the constraints imply $y_P\geq\epsilon\geq0$ for all $P$.
Its associated edge vector satisfies
$x(E)=\sum_P\lvert P\rvert y_P=a_f(G)\sum_P[n(P)-1]y_P=\gamma(E)$.
Thus Lemma~\ref{thm:Core} implies that it belongs to the core.
\end{proof}

\subsection{Equivalence of $LP'_2$ and $LP''_2$}
\label{sec:Equivalence}

Throughout this subsection, $\epsilon\geq0$.
Let $\boldsymbol{x}\in \mathcal{C}(\Gamma_G)$ and $\boldsymbol{y}\in \mathbb{R}^{\mathcal{E}}$ be a pair of associated vectors.
We show that $\boldsymbol{x}$ is feasible to $LP'_2$ if and only if $\boldsymbol{y}$ is feasible to $LP''_2$.
By Corollary \ref{thm:NonPrimeSet}, it is trivial that \eqref{eq:Nucleolus_LP2v2_4} and \eqref{eq:Nucleolus_LP2v3_4} are equivalent.
The equivalence of \eqref{eq:Nucleolus_LP2v2_3} and \eqref{eq:Nucleolus_LP2v3_3} follows from Lemma \ref{thm:PrimeSet_PartialOrder} and the observation below.

\begin{lemma}
Let $\boldsymbol{x}$ be a core allocation of $\Gamma_G$ and $P,Q$ be two prime sets of $G$ such that $Q$ is an ancestor of $P$.
Then $x_e \leq x_f$ for any $e\in P$ and any $f\in Q$. 
\end{lemma}

\begin{proof}
Let $\boldsymbol{x}^H$ be the vector associated with a densest subgraph $H$ of $G$.
Lemma \ref{thm:AncestorInclusion} implies that $Q\subseteq E(H)$ if $P\subseteq E(H)$.
It follows that for any $e\in P$ and any $f\in Q$, $x^H_e=x^H_f=\frac{1}{n(H)-1}$ if $P\subseteq E(H)$ and $x^H_e=0\leq x^H_f$ otherwise.
Since any vector $\boldsymbol{x}\in \mathcal{C}(\Gamma_G)$ is a convex combination of vectors associated with densest subgraphs of $G$, we have $x_e \leq x_f$ for any $e\in P$ and any $f\in Q$.
\end{proof}

Next, we show the equivalence of \eqref{eq:Nucleolus_LP2v2_2} and \eqref{eq:Nucleolus_LP2v3_2}.
Notice that \eqref{eq:Nucleolus_LP2v2_2} provides a characterization for trees in $\mathcal{T}_0$ with $\boldsymbol{x}\in \mathbb{R}^{E}$.
Thus \eqref{eq:Nucleolus_LP2v3_2} serves the same purpose.
To associate trees in $\mathcal{T}_0$ with $\boldsymbol{y}\in \mathbb{R}^{\mathcal{E}}$,
we introduce the following lemma.

\begin{lemma}
\label{thm:TightSpanningTree}
A spanning tree $T$ belongs to $\mathcal{T}_0$ if and only if for any prime set $P$ we have
\begin{equation}
\label{eq:TightSpanningTree}
\lvert T\cap P\rvert=n(P)-1.
\end{equation}
\end{lemma}

\begin{proof}
$(\Leftarrow)$
Assume that $T$ is a spanning tree satisfying \eqref{eq:TightSpanningTree} for any prime set $P$.
Let $H$ be a densest subgraph of $G$.
The vector $\boldsymbol{x}^H$ associated with $H$ is defined by $x^H_e=\frac{1}{n(H)-1}$ if $e\in E(H)$ and $x^H_e=0$ otherwise.
By Lemma \ref{thm:PrimePartition}, we have $E(H)=\cup_{i=1}^{r} P_i$ and $n(H)=\sum_{i=1}^{r} [n (P_i)-1]+1$, where $P_i\in \mathcal{P}$ for $i=1,\ldots,r$.
Thus
\begin{equation*}
x^H (T)=\frac{\sum_{i=1}^{r} \lvert T\cap P_i\rvert}{n(H)-1}=\frac{\sum_{i=1}^{r} [n (P_i)-1]}{n(H)-1}=\frac{n(H)-1}{n(H)-1}=1.
\end{equation*}
Since every vector in $\mathcal{C}(\Gamma_G)$ is a convex combination of vectors associated with densest subgraphs of $G$, we have $x(T)=1$ for any $\boldsymbol{x} \in \mathcal{C}(\Gamma_G)$, implying that $T\in \mathcal{T}_0$.

$(\Rightarrow)$
Assume that $T\in \mathcal{T}_0$, i.e., $T$ is a spanning tree such that $x(T)=1$ for any $\boldsymbol{x}\in \mathcal{C}(\Gamma_G)$.
Among all densest subgraphs of $G$ containing $P$, let $H$ be a minimal one.
We apply induction on the level of prime set $P\in \mathcal{P}$.

First assume that $P \in \mathcal{P}_0$.
Hence $P=E(H)$, implying that $H$ is a densest subgraph of $G$.
Then the vector $\boldsymbol{x}^H\in \mathcal{C}(\Gamma_G)$ associated with $H$ is defined by $x^H_e=\frac{1}{n(H)-1}$ for $e\in P$ and $x^H_e=0$ otherwise.
Since $x^H(T)=\frac{1}{n(H)-1} \lvert T\cap P\rvert = 1$, it follows that $\lvert T\cap P\rvert=n(H)-1=n (P)-1$.
Hence \eqref{eq:TightSpanningTree} holds for $P$.

Now assume that $P\in \mathcal{P}_k$, where $k\geq 1$.
Lemma \ref{thm:PrimePartition} implies that there exist prime sets $Q_1,\ldots,Q_r$ such that $E(H)=P\cup (\cup_{i=1}^{r} Q_i )$.
The minimality of $H$ implies that $P\prec Q_i$ for $i=1,\ldots,r$.
By the induction hypothesis, we have $\lvert T\cap (\cup_{i=1}^{r} Q_i)\rvert=\sum_{i=1}^{r} [n (Q_i)-1]$.
Then the vector $\boldsymbol{x}^H\in \mathcal{C}(\Gamma_G)$ associated with $H$ is defined by $x^H_e=\frac{1}{n(H)-1}$ if $e\in P\cup (\cup_{i=1}^{r} Q_i)$ and $x^H_e=0$ otherwise.
Since
\begin{equation*}
    x^H (T) = \frac{\lvert T\cap [P\cup(\cup_{i=1}^{r} Q_i)]\rvert}{n(H)-1} = \frac{\lvert T\cap P\rvert+\sum_{i=1}^{r} [n (Q_i)-1]}{n(H)-1} =1,
\end{equation*}
Lemma \ref{thm:PrimePartition} implies that $\lvert T\cap P\rvert=n(H)-1-\sum_{i=1}^{r} [n (Q_i)-1]=n (P)-1$.
Hence \eqref{eq:TightSpanningTree} holds for $P\in \mathcal{P}_k$ where $k\geq 1$.
\end{proof}

Now we are ready to prove the equivalence of \eqref{eq:Nucleolus_LP2v2_2} and \eqref{eq:Nucleolus_LP2v3_2}.

\begin{lemma}
Let $\boldsymbol{x}\in \mathcal{C}(\Gamma_G)$ be a vector and $\boldsymbol{y}\in \mathbb{R}^{\mathcal{E}}$ be the vector defined from $\boldsymbol{x}$.
Then $\boldsymbol{x}$ satisfies \eqref{eq:Nucleolus_LP2v2_2} if and only if $\boldsymbol{y}$ satisfies \eqref{eq:Nucleolus_LP2v3_2}.
\end{lemma}
\begin{proof}
Consider a spanning tree $T\in \mathcal{T}_0$ and its decomposition from the prime partition.
Lemma \ref{thm:TightSpanningTree} implies that
\begin{equation}
x(T)=\sum_{e\in T}x_e=\lvert T\cap E_0\rvert y_{E_0} + \sum_{P\in \mathcal{P}} \lvert T\cap P\rvert y_P=\sum_{P\in \mathcal{P}} [n (P)-1] y_P=1.
\end{equation}
Thus $\boldsymbol{x}$ satisfies \eqref{eq:Nucleolus_LP2v2_2} if and only if $\boldsymbol{y}$ satisfies \eqref{eq:Nucleolus_LP2v3_2}.
\end{proof}

Finally, we come to the equivalence of $LP'_2$ and $LP''_2$.
We first show that any vector $\boldsymbol{y}\in \mathbb{R}^{\mathcal{E}}$ defined from a feasible solution $\boldsymbol{x}\in \mathbb{R}^{E}$ of $LP'_2$ satisfies \eqref{eq:Nucleolus_LP2v3_1} and hence is a feasible solution of $LP''_2$.
To this end, we construct a spanning tree containing a prescribed number of edges from each prime set.
We need the following lemma.

\begin{lemma}
\label{thm:EdgeExchange}
Let $P,Q\in \mathcal{P}$ be a pair of prime sets with $P\prec Q$ in $(\mathcal{P},\prec)$.
For any spanning tree $T\in \mathcal{T}_0$, there exists a spanning tree $T'\in \mathcal{T}\backslash \mathcal{T}_0$ such that
\vspace{-.5em}
\begin{align}
\lvert T'\cap P\rvert&=\lvert T\cap P\rvert +1, \label{eq:EdgeExchange_1}\\
\lvert T'\cap Q\rvert&=\lvert T\cap Q\rvert -1, \label{eq:EdgeExchange_2}\\
\lvert T'\cap R\rvert&=\lvert T\cap R\rvert, \quad\qquad \forall ~R\in \mathcal{P}\backslash \{P,Q\}.  \label{eq:EdgeExchange_3}
\end{align}
\end{lemma}

\begin{proof}
Let $B=E\backslash E_0$, $\rho=a_f(G)$, and $b_R=n(R)-1=\lvert R\rvert/\rho$ for $R\in\mathcal{P}$.
Each nontrivial component of $G[B]$ is densest, so $r(B)=\sum_{R\in\mathcal{P}}b_R$.
For a subcollection $\mathcal{J}\subseteq\mathcal{P}$, write $B_{\mathcal{J}}=\bigcup_{R\in\mathcal{J}}R$ and $b(\mathcal{J})=\sum_{R\in\mathcal{J}}b_R$.
The definition of fractional arboricity gives $b(\mathcal{J})\le r(B_{\mathcal{J}})$.
Define integers $z_P=b_P+1$, $z_Q=b_Q-1$, and $z_R=b_R$ otherwise.
We claim that $z(\mathcal{J})\le r(B_{\mathcal{J}})$ for every $\mathcal{J}$.
Only the case $P\in\mathcal{J}$, $Q\notin\mathcal{J}$ needs consideration.
If $b(\mathcal{J})=r(B_{\mathcal{J}})$, then each nontrivial component of $G[B_{\mathcal{J}}]$ is densest.
By Lemma~\ref{thm:NonCrossingPrimeSet}, a component meeting $P$ contains all of $P$, and by Lemma~\ref{thm:AncestorInclusion} it also contains $Q$, a contradiction.
Thus $r(B_{\mathcal{J}})-b(\mathcal{J})\ge1$, proving the claim.

The matroid intersection theorem \cite{Schr03}, applied to the graphic matroid on $B$ and the partition matroid with capacities $z_R$, now gives a spanning forest $F$ of $G[B]$ with $\lvert F\cap R\rvert=z_R$ for all $R\in\mathcal{P}$.
To apply the matroid intersection theorem, it suffices to verify the following rank inequality.
For any $A\subseteq B$, let $\mathcal{J}=\{R:\lvert R\backslash A\rvert<z_R\}$.
Then
\begin{equation*}
r(A)+\sum_R\min\{z_R,\lvert R\backslash A\rvert\}
\ge r(B_{\mathcal{J}})-\lvert B_{\mathcal{J}}\backslash A\rvert+\lvert B_{\mathcal{J}}\backslash A\rvert+z(\mathcal{P}\backslash\mathcal{J})
\ge z(\mathcal{P})=r(B).
\end{equation*}
By Lemma~\ref{thm:TightSpanningTree}, $T\cap B$ is also a spanning forest of $G[B]$.
Consequently, $T'=F\cup(T\cap E_0)$ is a spanning tree of $G$.
The numbers of its edges in the prime sets satisfy \eqref{eq:EdgeExchange_1}--\eqref{eq:EdgeExchange_3}, and $T'\notin\mathcal{T}_0$ by Lemma~\ref{thm:TightSpanningTree}.
\end{proof}

\begin{lemma}
\label{thm:EquivalenceSufficiency}
Let $\boldsymbol{x}\in \mathcal{C}(\Gamma_G)$ be a vector satisfying \eqref{eq:Nucleolus_LP2v2_2}-\eqref{eq:Nucleolus_LP2v2_4} and $\boldsymbol{y}\in \mathbb{R}^{\mathcal{E}}$ be the vector defined from $\boldsymbol{x}$.
If $\boldsymbol{x}$ satisfies \eqref{eq:Nucleolus_LP2v2_1},
then $\boldsymbol{y}$ satisfies \eqref{eq:Nucleolus_LP2v3_1}.
\end{lemma}

\begin{proof}
Let $T\in \mathcal{T}_0$ be a spanning tree.
Lemma \ref{thm:EdgeExchange} implies that there exists a spanning tree $T'\in \mathcal{T}\backslash \mathcal{T}_0$ such that $\lvert T'\cap P\rvert=\lvert T\cap P\rvert +1$,
$\lvert T'\cap Q\rvert=\lvert T\cap Q\rvert -1$, and 
$\lvert T'\cap R\rvert=\lvert T\cap R\rvert$ for any $R\in \mathcal{P}\backslash \{P,Q\}$.
It follows that
\begin{align*}
x(T')
&=\lvert T'\cap P\rvert \cdot y_P +\lvert T'\cap Q\rvert \cdot y_Q + \sum_{R\in \mathcal{P}\backslash \{P,Q\}} \lvert T'\cap R\rvert \cdot y_R\\
&=(\lvert T\cap P\rvert + 1) \cdot y_P + (\lvert T\cap Q\rvert - 1) \cdot y_Q + \sum_{R\in \mathcal{P}\backslash \{P,Q\}} \lvert T\cap R\rvert \cdot y_R\\
&= x(T)+y_P-y_Q.
\end{align*}
Since $T\in \mathcal{T}_0$ and $T'\in \mathcal{T}\backslash \mathcal{T}_0$, we have $x(T)=1$ and $x(T')+\epsilon\leq 1$.
Hence $y_P + \epsilon \leq y_Q$ follows.
\end{proof}

Now we prove the converse implication.
By Lemma \ref{thm:OrderConeBound}, every feasible solution of $LP''_2$ defines a core allocation.
The following lemma uses \eqref{eq:OrderConeBound} to establish the remaining spanning-tree inequalities.

\begin{lemma}
\label{thm:EquivalenceNecessity}
Let $\boldsymbol{x}\in \mathcal{C}(\Gamma_G)$ be a vector satisfying \eqref{eq:Nucleolus_LP2v2_2}-\eqref{eq:Nucleolus_LP2v2_4} and $\boldsymbol{y}\in \mathbb{R}^{\mathcal{E}}$ be the vector defined from $\boldsymbol{x}$.
If $\boldsymbol{y}$ satisfies \eqref{eq:Nucleolus_LP2v3_1},
then $\boldsymbol{x}$ satisfies \eqref{eq:Nucleolus_LP2v2_1}.
\end{lemma}

\begin{proof}
Let $T\in\mathcal{T}\backslash\mathcal{T}_0$.
By Corollary~\ref{thm:CoreCoverHull} and the definition of $\mathcal{T}_0$, there is a densest subgraph $H$ such that
\begin{equation*}
d=n(H)-1-\lvert T\cap E(H)\rvert\geq1.
\end{equation*}
For each prime set $P$, put $w_P=y_P-\epsilon$ if $P\subseteq E(H)$ and $w_P=y_P$ otherwise, and put $w_{E_0}=0$.
The constraints on $\boldsymbol{y}$ imply $y_P\geq\epsilon$ for every prime set, so $w_P\geq0$.
If $P\prec Q$, Lemma~\ref{thm:AncestorInclusion} rules out $P\subseteq E(H)$ and $Q\not\subseteq E(H)$.
When $P\not\subseteq E(H)$ and $Q\subseteq E(H)$, \eqref{eq:Nucleolus_LP2v3_1} gives $w_P=y_P\leq y_Q-\epsilon=w_Q$; in the other cases, $w_P\leq w_Q$ follows directly.
Hence \eqref{eq:OrderConeBound} applies to $\boldsymbol{w}$.
Using Lemma~\ref{thm:PrimePartition} and the normalization constraint, we obtain
\begin{equation*}
1-x(T)
=\left(\sum_{P\in\mathcal{P}}[n(P)-1]w_P-w(T)\right)+\epsilon d
\geq\epsilon,
\end{equation*}
as required.
\end{proof}

Together with Lemma \ref{thm:OrderConeBound} and the preceding equivalences, Lemmas \ref{thm:EquivalenceSufficiency} and \ref{thm:EquivalenceNecessity} show that feasible solutions of $LP'_2$ and $LP''_2$ correspond under $x_e=y_P$ for $e\in P$ and $x_e=0$ for $e\in E_0$.
Since both programs maximize $\epsilon$, they have the same optimal value and corresponding optimal allocations.

\subsection{A combinatorial algorithm for $LP''_2$}
\label{sec:CombinatorialAlgorithm}

We first show that the optimal value of $LP''_2$ is positive.
Let $y_P$ be positive integers such that $y_Q\geq y_P+1$ whenever $P\prec Q$, and let $y_{E_0}=0$ and $\epsilon=1$.
Such integers exist because $\prec$ is a strict partial order.
Scaling $\boldsymbol{y}$ and $\epsilon$ to satisfy \eqref{eq:Nucleolus_LP2v3_2} gives a feasible solution with $\epsilon>0$.
The following lemma reveals how to solve $LP''_2$.

\begin{lemma}
\label{thm:TightConstraint}
Let $(\boldsymbol{y}^*,\epsilon^*)$ be an optimal solution of $LP''_2$.
Then for each prime set $P\in \mathcal{P}$, either $P\in\mathcal{P}_{\min}$ and $y_P^*=\epsilon^*$, or $y_R^*+\epsilon^*=y_P^*$ for some $R\prec P$.
\end{lemma}

\begin{proof}
Assume to the contrary that all the stated lower bounds on $y_{P_0}$ are slack for some $P_0\in\mathcal{P}$.
For a constant $\delta>0$ small enough, define $\boldsymbol{y}^{\star}$ by  $y^{\star}_{P_0}=y^*_{P_0}-\delta$ and $y^{\star}_{P}=y^*_{P}$ for any $P\in \mathcal{P}\backslash \{P_0\}$.
Then $(\boldsymbol{y}^{\star},\epsilon^*)$ satisfies \eqref{eq:Nucleolus_LP2v3_1}, \eqref{eq:Nucleolus_LP2v3_3} and \eqref{eq:Nucleolus_LP2v3_4}, but $\sum_{P\in \mathcal{P}} [n (P)-1] y^{\star}_P<1$.
Hence $(\boldsymbol{y}^{\star},\epsilon^*)$ can be scaled up with a constant $\theta>1$ such that $(\theta \boldsymbol{y}^{\star},\theta \epsilon^*)$ satisfies \eqref{eq:Nucleolus_LP2v3_1}-\eqref{eq:Nucleolus_LP2v3_4}.
However, this contradicts the optimality of $(\boldsymbol{y}^*,\epsilon^*)$.
\end{proof}

Based on the lemma above, we derive a combinatorial algorithm for solving $LP''_2$.

\begin{algorithm}[H]
\renewcommand{\thealgorithm}{}
\setstretch{1.15}
\caption{A combinatorial algorithm for $LP''_2$}
\label{alg}

\begin{algorithmic}[1]
\State $k=0$, $\mathcal{R}\gets\mathcal{P}$, $y_{E_0}\gets0$
    \While{$\mathcal{R}\not=\emptyset$}
    \State $k \gets k+1$
    \State $\mathcal{R}_{\min} \gets $ the set of minimal prime sets in $(\mathcal{R},\prec)$
    \State $y_P \gets k \epsilon$ for any $P\in \mathcal{R}_{\min}$
    \State $\mathcal{R} \gets \mathcal{R}\backslash \mathcal{R}_{\min}$
    \EndWhile
\State Since $y_P=k_P \epsilon$ where $k_P$ is an integer for any $P\in \mathcal{P}$, solving for $\epsilon$ in \eqref{eq:Nucleolus_LP2v3_2} gives the unique optimal solution of $LP''_2$.
\end{algorithmic}
\end{algorithm}

Let $k_P$ be the iteration in which $P$ is removed.
Every feasible solution satisfies $y_P\geq k_P\epsilon$, by induction along the strict order.
Thus $\epsilon\leq(\sum_P[n(P)-1]k_P)^{-1}$.
The algorithm attains this bound, since $k_Q\geq k_P+1$ whenever $P\prec Q$.
As all coefficients $n(P)-1$ are positive, equality in the normalization forces $y_P=k_P\epsilon$ for every $P$ at the optimum.
Therefore $LP''_2$ has a unique optimal solution.
By Lemma \ref{thm:SecondRoundFormulation} and the equivalence of $LP'_2$ and $LP''_2$, its associated allocation is the nucleolus of $\Gamma_G$.
Now we are ready to present our main result.

\begin{theorem}
\label{thm:Nucleolus}
Let $\Gamma_G=(N,\gamma)$ be an arboricity game with a nonempty core.
The nucleolus of $\Gamma_G$ can be computed in $O(n^4 m \log(2+\frac{n^2}{m}))$.
\end{theorem}

\begin{proof}
By Theorem \ref{thm:PrimePartitionComplexity}, the prime partition can be computed in $O(n^4 m \log(2+\frac{n^2}{m}))$.
By Lemma \ref{thm:PrimeSetRelation}, the ancestor relation of prime sets can be determined in $O(n^2 m)$.
The algorithm above takes $O(n)$ iterations and each iteration requires $O(n^2)$ time to determine the minimal prime sets in the remaining partially ordered set.
Hence the algorithm above runs in $O(n^3)$.
Notice that the prime partition computation dominates the computing time of all other parts.
Thus the nucleolus can be computed in $O(n^4 m \log(2+\frac{n^2}{m}))$.
\end{proof}

The following example shows that the algorithm need not compute the nucleolus when coalition costs are defined by fractional arboricity.

\begin{example}
\label{ex:FractionalVariant}
Let $G$ consist of two parallel $ab$ edges and the edges $ac,bc$.
In this edge order, the algorithm returns $(2/3,2/3,1/3,1/3)$, whereas the game with costs defined by fractional arboricity has nucleolus $(3/4,3/4,1/4,1/4)$.
Indeed, its core consists of $(t,t,1-t,1-t)$ for $1/2\leq t\leq1$, and the second-round objective is $\min\{1-t,t-1/2\}$, uniquely maximized at $t=3/4$.
\end{example}

\section{Concluding remarks}
\label{sec:Conclusion}

This paper provides an efficient algorithm for computing the nucleolus of arboricity games when the core is not empty.
Notice that a variant of the arboricity game arises when the cost of each coalition is defined by fractional arboricity instead of arboricity.
This variant always has a nonempty core, but Example \ref{ex:FractionalVariant} shows that the algorithm above does not in general compute its nucleolus.

This paper also offers a graph decomposition built on the lattice of maximum-density edge sets.
The prime partition decomposes the edge set of a graph into a non-prime set and a number of prime sets, where prime sets correspond to minimal densest minors.
Notice that the non-prime set can be further decomposed following the same procedure for defining prime sets.
Therefore, the prime partition indeed provides a hierarchical graph decomposition analogous to the core decomposition.

\section*{Acknowledgments}
We are grateful to all reviewers for their comments and suggestions, which greatly improved the presentation of this work.
The first author is supported in part by the National Natural Science Foundation of China (No.\,12001507) and the Shandong Provincial Natural Science Foundation (No.\,ZR2020QA024).
The second author is supported in part by the National Natural Science Foundation of China (Nos.\,11871442 and 11971447).

\bibliographystyle{abbrv}
\bibliography{reference}

\end{document}